\documentclass[11pt]{article}

\usepackage[a4paper]{geometry}
\usepackage{subcaption}

\usepackage{amssymb}
\usepackage{amsmath,amsfonts,amscd,amsthm,xspace}
\usepackage{authblk}
\usepackage{epsfig}
\usepackage{wrapfig}
\usepackage{enumitem}
\usepackage{dblfloatfix}
\usepackage{array}
\usepackage{mathtools}
\usepackage{subcaption}
\usepackage{fixltx2e}
\usepackage{color}
\usepackage[utf8]{inputenc}
\usepackage[ruled]{algorithm2e}
\usepackage[T1]{fontenc}
\usepackage{url}
\usepackage{booktabs}
\usepackage{amsfonts}
\usepackage{nicefrac}
\usepackage{microtype}   
\usepackage{graphicx}
\graphicspath{{Figs/}}

\usepackage{newtxtext}

\newtheorem{theorem}{Theorem}
\numberwithin{theorem}{section}
\newtheorem{lemma}{Lemma}
\newtheorem*{nonumbertheorem}{Theorem}

\numberwithin{corollary}{section}

\theoremstyle{definition}
\newtheorem{definition}{Definition}
\numberwithin{definition}{section}

\newcommand{\nn}{n}
\newcommand{\nm}{m}

\newcommand{\nt}{T}

\DeclareMathOperator*{\argmin}{arg\,min}

\makeatletter
\newcommand{\printfnsymbol}[1]{%
  \textsuperscript{\@fnsymbol{#1}}%
}
\makeatother
  
\date{}

\begin{document}

\title{Robustness and Consistency in Linear Quadratic Control with Untrusted Predictions}

\author[1]{
  Tongxin~Li\thanks{denotes equal contribution}}
\author[2]{Ruixiao Yang \printfnsymbol{1}} \author[3]{Guannan Qu} \author[1]{Guanya Shi} \author[2]{Chenkai Yu} \author[1]{Adam Wierman} \author[1]{Steven Low}
\affil[1]{California Institute of Technology}
\affil[2]{Tsinghua University}
\affil[3]{Carnegie Mellon University}

\maketitle
\begin{abstract}
    We study the problem of learning-augmented predictive linear quadratic control. Our goal is to design a controller that balances \textit{``consistency''}, which measures the competitive ratio when predictions are accurate, and \textit{``robustness''}, which bounds the competitive ratio when predictions are inaccurate. 
    We propose a novel $\lambda$-confident policy and provide a competitive ratio upper bound
    that depends on a trust parameter $\lambda\in [0,1]$ set based on the confidence in the predictions and some prediction error $\varepsilon$.
 Motivated by online learning methods, we design a self-tuning policy that adaptively learns the trust parameter $\lambda$ with a competitive ratio that depends on $\varepsilon$ and the variation of system perturbations and predictions. We show that its competitive ratio is  bounded from above by $ 1+{O(\varepsilon)}/({{\Theta(1)+\Theta(\varepsilon)}})+O(\mu_{\mathsf{Var}})$ where $\mu_\mathsf{Var}$ measures the variation of perturbations and predictions. It implies that when the variations of perturbations and predictions are small, by automatically adjusting the trust parameter online, the self-tuning scheme ensures a competitive ratio that does not scale up with the prediction error $\varepsilon$.
\end{abstract}

We study a classical online linear quadratic control problem where the controller has access to untrusted predictions/advice during each round, potentially from a black-box AI tool. 

One consequence of the success of machine learning is that accurate predictions are available for many online decision and control problems. For example, the development of deep learning has enabled generic prediction models in various domains, e.g. weather, demand forecast, and user behaviors. Such predictions are powerful because future information plays a significant role in optimizing the current control decision. The availability of accurate future predictions can potentially lead to order-of-magnitude performance improvement in decision and control problems, where one can simply plug-in the predictions and achieve near optimal performance when compared to the best control actions in hindsight, a.k.a., \emph{consistency}.
However, an important caveat is that the predictions are helpful only when they are accurate, which is not guaranteed in many scenarios. Since many predictions are obtained from black box AI models like neural networks, there is no uncertainty quantification and it is unclear whether the predictions are accurate.
In the case when the predictions are not accurate, the consequences can be catastrophic, leading to unbounded worst-case performance, e.g., an unbounded competitive ratio. The possibility of such worst-case unbounded competitive ratio prevents the use of ML predictions in safety-critical applications that are adverse to potential risks. 

The use of predictions described above is a sharp contrast to the approaches developed by the online algorithm community, where the algorithms have access to no future prediction, yet can be \emph{robust} to all future variations and achieve a finite competitive ratio. While such algorithms miss out on the improvements possible when accurate predictions are available, their robustness properties are necessary in safety-critical settings. Therefore, a natural question arises: 
\begin{center}
    \emph{Can such adversarial guarantees be provided for control policies that use black-box AI predictions?}
\end{center}

To provide adversarial guarantees necessarily means not precisely following the black-box AI predictions.  Thus, there must be a trade-off between the performance in the typical case (consistency) and the quality of the adversarial guarantee (robustness).  Trade-offs between consistency and robustness have received considerable attention in recent years in the online algorithms community, starting with the work of \cite{lykouris2018competitive}, but our work represents the first work in the context of control.

\textbf{Contributions.}  In this paper, we answer the question above in the affirmative, in the context of linear quadratic control, providing an novel algorithm that trades off consistency and robustness to provide adversarial gaurantees on the use of untrusted predictions.

Our first result provides a novel online control algorithm, termed \textit{$\lambda$-confident control}, that provides a competitive ratio of 
$1+\min\{O(\lambda^2\varepsilon)+ O(1-\lambda)^2,O(1)+O(\lambda^2)\}$, where $\lambda\in [0,1]$ is a \textit{trust parameter} set based on the confidence in the predictions, and $\varepsilon$ is the prediction error (Theorem~\ref{thm:upper_bound}). When the predictions are accurate ($\varepsilon\approx 0$), setting $\lambda$ close to $1$ will obtain a competitive ratio close to $1$, and hence the power of the predictions is fully utilized; on the other hand, when the predictions are inaccurate ($\varepsilon$ very large), setting $\lambda\approx 0$ will still guarantee a constant competitive ratio, meaning the algorithm will still have good robustness guarantees when the predictions turn out to be bad. Therefore, our approach can get the best of both worlds, effectively using black-box predictions but still guaranteeing robustness.  

The above discussion highlights that the optimal choice of $\lambda$ depends on the prediction error, which may not be known a priori.  
Therefore, we further provide an adaptive, self-tuning learning policy (Algorithm~\ref{alg:self}) that selects $\lambda$ so as to learn the optimal parameter for the actual prediction error; thus selecting the optimal balance between robustness and consistency.
Our main result proves that the self-tuning policy maintains a competitive ratio bound that is always bounded regardless of the prediction error $\varepsilon$ (Theorem~\ref{thm:competitive_self_tuning}). This result is informally presented below. 
\begin{nonumbertheorem}[Informal]
\label{thm:competitive_self_tuning_informal}
Under our model assumptions, there is a  self-tuning online control algorithm that selects some $\lambda_t\in [0,1]$ for all $t=0,\ldots,\nt-1$ and achieves a
competitive ratio
\begin{align*}
\mathsf{CR}(\varepsilon) \leq 1+\frac{O(\varepsilon)}{{\Theta(1)+\Theta(\varepsilon)}}+O(\mu_{\mathsf{Var}})
\end{align*} as a function of the prediction error $\varepsilon$ where $\mu_\mathsf{Var}$ measures the variation of perturbations and predictions.
\end{nonumbertheorem}


This result provides a worst-case performance bound for the use of untrusted predictions, e.g., the predictions from a black-box AI tool, regardless of the accuracy of the predictions.   The second term in the competitive ratio upper bound indicates a nontrivial non-linear dependency of $\mathsf{CR}(\varepsilon)$ and the prediction error $\varepsilon$, matching our experimental results shown in Section~\ref{sec:experiments}. The third term measures the variation of perturbations and predictions. Such a term is common in regret analysis based on the ``Follow The Leader'' (FTL) approach~\cite{hannan20164,kalai2005efficient}. For example, the regret analysis of the Follow the Optimal Steady State (FOSS) method in~\cite{li2019online} contains a similar  ``path length'' term that captures the variation of the state trajectory.


Proving our main result is complex due to the fact that, different from classical online learning models, the cost function in our problem depends on previous actions via a linear dynamical system (see~\eqref{eq:linear_dynamic}). The time coupling can even be exponentially large if the dynamical system is unstable. To tackle this time-coupling structure, we develop a new proof technique that relates the regret and competitive ratio with the convergence rate of the trust parameter. 

Finally, in Section \ref{sec:experiments} we demonstrate the effectiveness of our self-tuning approach using three examples: a robotic tracking problem, an adaptive battery-buffered EV charging problem and the Cart-Pole problem. For the robotic tracking and adaptive battery-buffered EV charging cases, we illustrate that the competitive ratio of the self-tuning policy performs nearly as well as the lower envelope formed by picking multiple trust parameters optimally offline. We also validate the practicality of our self-tuning policy by showing that it not only works well for linear quadratic control problems; it also performs well in the nonlinear Cart-Pole problem.



\textbf{Related Work.} Our work contributes to the growing literature on learning-augmented online algorithm design. There has been significant interest in the goal of trading-off consistency and robustness in order to ensure worst-case performance bounds for black-box AI tools in online problems. As discussed earlier, prediction based algorithms can achieve consistency, while online algorithms can have robustness. These two classes of algorithms can be viewed as two extremes, and a number of works attempt to develop algorithms that balance between consistency and robustness in settings like online caching \cite{lykouris2018competitive}, ski-rental \cite{angelopoulos2019online,purohit2018improving,wei2020optimal,bamas2020primal}, online set cover \cite{bamas2020primal}, secretary
and online matching \cite{antoniadis2020secretary}, and metric task systems \cite{antoniadis2020online}. For example, in the ski rental problem, \cite{purohit2018improving} proposes an algorithm that achieves $1+\lambda$ consistency and $1 + \frac{1}{\lambda}$ robustness for a tuning parameter $\lambda\in(0,1)$. 
Compared to these works, our setting is fundamentally more challenging because of the existence of dynamics in the control problem couples all decision points, and a mistake at one time can be magnified and propagated to all future time steps. 

Our work is also closely related to a broad literature on regret and competitive ratio analysis for Linear Quadratic Control (LQC) and Linear Quadratic Regulator (LQR) systems with predictions. In~\cite{yu2020power}, LQR regret analysis for Model Predictive Control (MPC) is given, assuming accurate predictions of perturbations. Inaccurate predictions are considered in~\cite{yu2020competitive}, with competitive results provided. It is proven in~\cite{yu2020power,yu2020competitive} that the action generated by MPC can be explicitly written as the action of optimal linear control plus a linear combination of inaccurate predictions. The competitive analysis of the consistent and robust control scheme in this work makes use of this fact in the analysis of the more challenging case of untrusted predictions. Other related regret and competitive ratio results for MPC include~\cite{NEURIPS2020_ed46558a,li2021information,zhang2021regret}.

While our work is the first to study learning-augmented control via the lens of robustness and consistency, there are two classical communities in control that are related to the goals of our work:  robust control and adaptive control.

Robust control is a large area that concerns the design of controllers with performance guarantees that are robust against model uncertainty or adversarial disturbances \cite{dullerud2013course}. Tools of robust control include $H_\infty$ synthesis \cite{doyle1988state,zhou1996robust} and robust MPC \cite{bemporad1999robust}. 
Like the robust control literature, our work also considers robustness, but our main focus is on balancing between robustness and consistency in a predictive control setting. Consistency is not a focus of the robust control literature.  Further, we focus on the metrics of competitive ratio and regret, which is different from the typical performance measures in the robust control literature, which focus on measures such as system norms \cite{zhou1996robust}. 

The design of our self-tuning control in Section \ref{sec:self-tuning-control} falls into the category of adaptive control. There is a rich body of literature studying Lyapunov stability and asymptotic convergence in adaptive control theory \cite{slotine1991applied}. Recently, there has been increasing interest in studying adaptive control with non-asymptotic metrics from learning theory. Typical results guarantee convergence in finite time horizons using measures such as regret \cite{abbasi2011regret,simchowitz2020naive,dean2019sample,agarwal2019online}, dynamic regret \cite{zhang2021regret,yu2020power,li2019online}, and competitive ratio \cite{NEURIPS2020_ed46558a,yu2020competitive}. Different from these works, this paper deploys an adaptive policy with the goal of balancing robustness and consistency.  Additionally, such results do not focus on incorporation of untrusted predictions.


    
    
    


\section{Preliminaries}
\label{sec:model}

We consider a Linear Quadratic Control (LQC) model. Throughout this paper, $\|\cdot\|$ denotes the $\ell_2$-norm for vectors and the matrix norm induced by the $\ell_2$-norm.  Denote by $x_t\in\mathbb{R}^{\nn}$ and $u_t\in\mathbb{R}^{\nm}$ the system state and action at each time $t$. We consider a linear dynamic system with adversarial perturbations,
\begin{align}
\label{eq:linear_dynamic}
    x_{t+1} = Ax_t + B u_t + w_t, \text{ for } t = 0,\ldots,\nt-1,
\end{align}
where $A\in\mathbb{R}^{\nn\times\nn}$ and $B\in\mathbb{R}^{\nn\times\nm}$, and $w_t\in\mathbb{R}^{\nn}$ denotes some unknown perturbation chosen adversarially. We make the standard assumption that the pair $(A,B)$ is stabilizable.
Without loss of generality, we also assume the system is initialized with some fixed $x_0\in\mathbb{R}^{n}$. The goal of control is to minimize the following quadratic costs given matrices $A,B,Q,R$ :
\begin{align*}
    J \coloneqq  \sum_{t=0}^{\nt-1} (x_t^\top Q x_t + u_t^\top R u_t) + x_{\nt} P x_{\nt},
\end{align*}
where $Q,R\succ 0$ are positive definite matrices, and $P$ is the solution of the following discrete algebraic Riccati equation (DARE), which must exist because $(A,B)$ is stabilizable and $Q,R\succ 0$~\cite{dullerud2013course}. 
\begin{align*}
    P=Q+A^\top P A - A^\top PB (R+B^\top P B)^{-1} B^\top PA.
\end{align*}
Given $P$,  we can define $K \coloneqq (R+B^\top P B)^{-1} B^\top P A$ as the optimal LQC controller in the case of no disturbance ($w_t =0$). Further, let $F\coloneqq A-BK$ be the closed-loop system matrix when using $u_t = -K x_t$ as the controller. By \cite{dullerud2013course}, $F$ must have a spectral radius $\rho(F)$ less than $1$. Therefore, Gelfand’s formula implies that there must exist a constant $C>0$, $\rho\in(0,1)$ s.t. $\left\Vert F^t\right\Vert\leq C\rho^t,\forall t\geq 0$.


Our model is a classical control model~\cite{khalil1996robust} and has broad applicability across various engineering fields. In the following, we introduce ML/AI predictions into the classical model and study the trade-off between consistency and robustness in this classical model for the first time. 

\subsection{Untrusted Predictions}

Our focus is on predictive control and we assume that, at the beginning of the control process, a sequence of predictions of the disturbances $\left(\widehat{w}_0,\ldots,\widehat{w}_{\nt-1}\right)$ is given to the decision maker. At time $t$, the decision maker observes $x_t, w_{t-1}$ and picks a decision $u_t$. Then, the environment picks $w_t$, and the system transitions to the next step according to \eqref{eq:linear_dynamic}. We emphasize that, at time $t$, the decision maker has no access to $(w_t,\ldots, w_T)$ and their values may be different from the predictions $(\widehat{w}_t,\ldots,\widehat{w}_T)$. Also, note that $w_t$ can be adversarially chosen at each time $t$, adaptively.

The assumption that a sequence of predictions available is $(\widehat{w}_t,\ldots, \widehat{w}_T)$ is not as strong as it may first appear, nor as strong as other similar assumptions made in literature, e.g.,~\cite{yu2020power,li2019online}, because we allow for prediction error. If there are no predictions or only a subset of predictions are available, we can simply set the unknown predictions to be zero and this does not affect our theoretical results and algorithms. 

In our model, there are two types of uncertainty. The first is caused by the \textit{perturbations} because the future perturbations $(w_t,\ldots,w_{\nt-1})$ are unknown to the controller at time $t$. The second is the \textit{prediction error} due to the mismatch $e_t\coloneqq  \widehat{w}_t-w_t$ between the perturbation $w_t$ and the prediction $\widehat{w}_t$ at each time. Formally, we define the prediction error as 
\begin{align}
\label{eq:prediction_error}
    \varepsilon\left(F,P,e_0,\ldots,e_{\nt-1}\right) & \coloneqq  \sum_{t=0}^{\nt-1}\left\|\sum_{\tau=t}^{\nt-1}\left(F^\top\right)^{\tau-t}P e_t\right\|^2.
\end{align}

Notice that the prediction error is not defined as a form of classical mean squared error for our problem. The reason is because the mismatch $e_t$ at each time has different impact on the system. Writing the prediction error as in~\eqref{eq:prediction_error} simplifies our analysis.  {In  fact, if we define $u_t$ and $\widehat{u}_t$ as two actions given by an optimal linear controller (formally defined in Section~\ref{sec:0_confident}) as if the true perturbations are $w_0,\ldots,w_{\nt-1}$ and $\widehat{w}_0,\ldots,\widehat{w}_{\nt-1}$, respectively, then it can be verified that $ \varepsilon=\sum_{t=0}^{\nt-1}\|u_t-\widehat{u}_t\|$, which is the accumulated action mismatch for an optimal linear controllers provided with different estimates of perturbations.} In Section~\ref{sec:experiments}, using experiments, we show that the competitive ratios (with a fixed ``trust parameter'', defined in~\ref{sec:confident_control}) grow linearly in the prediction error $\varepsilon$ defined in~\eqref{eq:prediction_error}.
Finally, we assume that the perturbations $\left(w_0,\ldots,w_{\nt-1}\right)$ and predictions $\left(\widehat{w}_0,\ldots,\widehat{w}_{\nt-1}\right)$ are uniformly bounded, i.e., there exist $\overline{w}>0$ and $\widehat{w}>0$ such that $\left\|w_t\right\|\leq \overline{w}$ and $\left\|\widehat{w}_t\right\|\leq \widehat{w}$ for all $0\leq t\leq \nt-1$. 
{In summary, Figure~\ref{fig:system} demonstrates the system model considered in this paper.}

\begin{figure}[ht]
\centering
\includegraphics[scale=0.47]{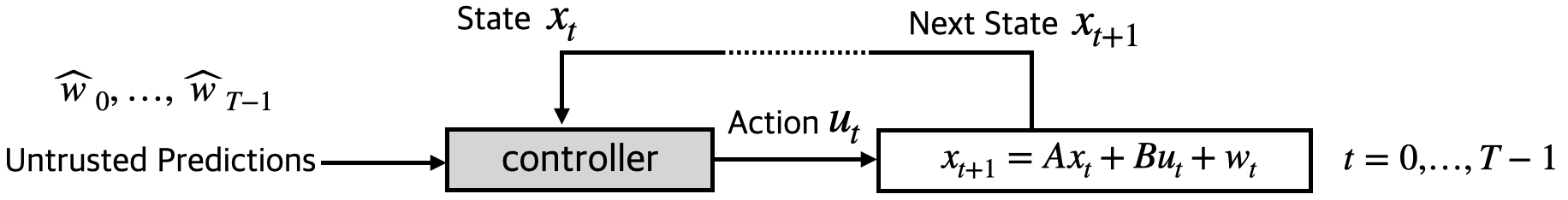}
\caption{System model of linear quadratic control (LQC) with untrusted predictions.}
\label{fig:system}
\end{figure}

\subsection{Defining Consistency and Robustness}
\label{sec:cr}

As discussed in the in introduction, while predictions can be helpful, inaccurate predictions can lead to unbounded competitive ratio.  Our goal is to utilize predictions to achieve good performance (consistency) while still providing adversarial worst-case guarantees (robustness). In this subsection, we formally define the notions of consistency and robustness we study.  These notions have received increasing attention recently in the area of online algorithms with untrusted advice, e.g., \cite{angelopoulos2019online,purohit2018improving,wei2020optimal,bamas2020primal,antoniadis2020secretary,antoniadis2020online}.

We use the competitive ratio to measure the performance of an online control policy and quantify its robustness and consistency. Specifically, let $\mathsf{OPT}$ be the offline optimal cost when all the disturbances $(w_t)_{t=0}^{\nt}$ are known, and $\mathsf{ALG}$ be the cost achieved by an online algorithm. Throughout this paper we assume $\mathsf{OPT}>0$. We define the competitive ratio for a given bound on the prediction error $\varepsilon$, as follows.

\begin{definition} \emph{
The \textbf{competitive ratio} for a given  prediction error $\varepsilon$, $\mathsf{CR}(\varepsilon)$, is defined as the smallest constant $C\geq 1$ such that $\mathsf{ALG}\leq C \cdot \mathsf{OPT}$ for fixed $A,B,Q,R$ and any adversarially and adaptively chosen perturbations $(w_0,\ldots,w_{\nt-1})$ and predictions $(\widehat w_0,\ldots,\widehat w_{\nt-1})$.} 
\end{definition}

Building on the definition of competitive ratio, we define robustness and consistency as follows.

\begin{definition} \emph{
An online algorithm is said to be \textbf{$\gamma$-robust} if,
for any prediction error $\varepsilon>0$, the competitive ratio satisfies $\mathsf{CR}(\varepsilon)\leq \gamma$, and an algorithm is said to be \textbf{$\beta$-consistent} if
the competitive ratio satisfies $\mathsf{CR}(0)\leq \beta$. }
\end{definition}





\subsection{Background: Existing Algorithms}

Before proceeding to our algorithm and its analysis, we first introduce two extreme algorithm choices that have been studied previously: a myopic policy that we refer to as $1$-confident control, which places full trust in the predictions, and a pure online strategy that we refer to as $0$-confident control, which places no trust in the predictions.  These represent algorithms that can achieve consistency and robustness \emph{individually}, but cannot achieve consistency and robustness \emph{simultaneously}.  The key challenge of this paper is to understand how to integrate ideas such as what follows into an algorithm that achieves consistency and robustness simultaneously.

\subsubsection{A Consistent Algorithm: $1$-Confident Control} A simple way to achieve consistency is to put full faith in the untrusted predictions.  In particular, if the algorithm trusts the untrusted predictions and follows them, the performance will always be optimal if the predictions are accurate.  We refer to this as the $1$-confident policy, which is defined by a finite-time optimal control problem that trusts that $(\widehat{w}_0,\ldots,\widehat{w}_{\nt-1})$ are the true disturbances. Formally, at time step $t$, the actions $(u_t,\ldots,u_{\nt})$ are computed via
\begin{align}
    \label{eq:MPC}
    &\argmin_{(u_t,\ldots,u_{\nt-1})} \left(\sum_{\tau=t}^{\nt-1} (x_\tau^\top Q x_\tau + u_\tau^\top R u_\tau) + x_\nt P x_\nt\right) \  \ \text{s.t.  }~\eqref{eq:linear_dynamic} \text{ for all } \tau=t,\ldots,\nt-1.
\end{align}
With the obtained solution $(u_t,\ldots,u_{\nt-1})$, the control action $u_t$ at time $t$ is fixed to be $u_t$ and the other actions $(u_{t+1},\ldots,u_{\nt-1})$ are discarded. 

We highlight the following result (Theorem 3.2 in~\cite{yu2020power}) that provides an explicit expression of the algorithm in~\eqref{eq:MPC}, which can be viewed as a form of Model Predictive Control (MPC).

\begin{theorem}[Theorem 3.2 in~\cite{yu2020power}]
With predictions $(\widehat{w}_0,\ldots,\widehat{w}_{\nt-1})$ fixed, the solution $u_{t}$ of the algorithm in~\eqref{eq:MPC} can be expressed as
\begin{align}
\label{eq:MPC_explicit}
    u_{t} =  -(R+B^{\top}PB)^{-1}B^{\top}\left(PAx_t+\sum_{\tau=t}^{\nt-1}\left(F^\top\right)^{\tau-t} P\widehat{w}_{\tau}\right)
\end{align}
where $F\coloneqq A-B(R+B^\top P B)^{-1} B^\top P A=A-BK$.
\end{theorem}

It is clear that this controller \eqref{eq:MPC} (or equivalently \eqref{eq:MPC_explicit}) achieves $1$-consistency because, when the prediction errors are $0$, the control action from \eqref{eq:MPC} (and the state trajectory) will be exactly the same as the offline optimal. However, this approach is not robust, and one can show that prediction errors can lead to unbounded competitive ratios. In the next subsection, we introduce a robust (but non consistent) controller.  

\subsubsection{A Robust Algorithm: $0$-Confident Control.} \label{sec:0_confident} On the other extreme, a natural way to be robust is to ignore the untrusted predictions entirely, i.e., place no confidence in the predictions.  The $0$-confident policy does exactly this.  It places no trust in the predictions and synthesizes the controller by assuming $w_t=0$. Formally, the policy is given by
\begin{align}
\label{eq:online_explicit}
u_t=-Kx_t.
\end{align}

This recovers the optimal pure online policy in classical linear control theory~\cite{anderson2012optimal}. As shown by \cite{yu2020power}, this controller has a constant competitive ratio and therefore is $O(1)$-robust. However, this approach is not consistent as it does not utilize the predictions at all. In the next section, we discuss our proposed approach, which achieves both consistency and robustness.

\section{Consistent and Robust Control}

The goal of this paper is to develop a controller that performs near-optimally when predictions are accurate (consistency) and meanwhile is robust when the prediction error is large. As discussed in the previous section, a myopic, $1$-confident controller that puts full trust into the predictions is consistent, but not robust.  On the other hand, any purely online $0$-confident policy that ignores predictions is robust but not consistent.  

The algorithms we present establish a trade-off between these extremes by including a  ``confidence/trust level'' for the predictions.  The algorithm design challenge is to determine the right way to balance these extremes.  In the first (warmup) algorithm, the policy starts out confident in the predictions, but when a threshold of error is observed, the policy loses confidence and begins to ignore predictions.  This simple threshold-based policy highlights that it is possible for a policy to be both robust and consistent.  However, the result also highlights the weakness of the standard notions of robustness and consistency since the policy cannot make use of intermediate quality predictions and only performs well in the extreme cases when predictions are either perfect or poor. 

Thus, we move to considering a different approach, which we term  \textit{$\lambda$-confident control}. This algorithm selects a confidence level $\lambda$ that serves as a weight for a linear combination between purely myopic $1$-confident control and purely online $0$-confident control.  Our main result shows that this policy provides a smooth trade-off between robustness and consistency and, further, in Section \ref{sec:self-tuning-control}, we show that the confidence level $\lambda$ can be learned online adaptively so as to achieve consistency and robustness without exogenously specifying a trust level. 

\subsection{Warmup: Threshold-Based Control}

We begin by presenting a simple threshold-based algorithm that can be both robust and consistent, though it does not perform well for predictions of intermediate quality.  This distinction highlights that looking beyond the classical narrow definitions of robustness and consistency is important when evaluating algorithms. 


\begin{algorithm}[t]
\SetAlgoLined
Initialize $\delta = 0$\\
\For{$t=0,\ldots,\nt-1$}{
\eIf{$\delta< \sigma$}{
$u_t=-(R+B^{\top}PB)^{-1}B^{\top}\left(PAx_t+\sum_{\tau=t}^{T-1}\left(F^{\top}\right)^{\tau-t}P\widehat{w}_\tau\right)$
}{
Compute $u_t$ with the best myopic online algorithm $\mathcal{A}_{\mathrm{Online}}$ without predictions
}
Update $x_{t+1}=A x_t+B u_t + w_t$ and $\delta \leftarrow \delta + \|\widehat{w}_t-w_t\|$
}
\caption{Threshold-Based Control}
\label{alg:threshold}
\end{algorithm}

The threshold-based algorithm is described in Algorithm \ref{alg:threshold}. It works by trusting predictions (using $1$-confident control update \eqref{eq:MPC_explicit}) until a certain error threshold $\sigma>0$ is crossed and then ignoring predictions (using an online algorithm $\mathcal{A}_{\mathrm{Online}}$ that attains a (minimal) competitive ratio $C_{\min}$\footnote{Note that $C_{\min}$ is guaranteed to exist, as setting $\lambda=0$ in Theorem~\ref{thm:upper_bound} gives a constant $1+\|H\|/\lambda_{\min}(G)$ competitive ratio bound for the $0$-confident control update~\eqref{eq:online_explicit}, therefore $1\leq C_{\min}\leq 1+\|H\|/\lambda_{\min}(G)$.} for all online algorithms that do not use predictions). 
The following result shows that, with a small enough threshold, this algorithm is both robust and consistent because, if predictions are perfect, it trusts them entirely, but if there is an error, it immediately begins to ignore predictions and matches the $0$-confident controller performance, which is optimal. A proof can be found in Appendix~\ref{appendix:proof_of_threshold}.

\begin{theorem} 
\label{thm:threshold}
There exists a threshold parameter $\sigma>0$ such that Algorithm~\ref{alg:threshold} is $1$-consistent and $(C_{\mathrm{min}}+o(1))$-robust, where $C_{\mathrm{min}}$ is the minimal competitive ratio of any pure online algorithm. 
\end{theorem}

The term $o(1)$ in Theorem~\ref{thm:threshold} converges to $0$ as $\nt\rightarrow\infty$.
While Algorithm \ref{alg:threshold} is optimally robust and consistent, it is unsatisfying because it does not improve over the online algorithm unless predictions are perfect since in the proof, we set the threshold parameter $\sigma>0$ arbitrarily small to make the algorithm robust and $1$-consistent and the definition of consistency and robustness only captures the behavior of the competitive ratio $\mathsf{CR}(\varepsilon)$ for either $\varepsilon=0$ or $\varepsilon$ is large. As a result, in the remainder of the paper we look beyond the extreme cases and prove results that apply for arbitrary prediction error quality.  In particular, we prove competitive ratio bounds that hold for arbitrary $\varepsilon$, of which consistency and robustness are then special cases.


\subsection{$\lambda$-Confident Control}
\label{sec:confident_control}

We now present our main results, which focus on a policy that, like Algorithm \ref{alg:threshold}, looks to find a balance between the two extreme cases of $1$-confident and $0$-confident control.  However, instead of using a threshold to decide when to swap between them, the $\lambda$-confident controller considers a linear combination of the two. 
\begin{algorithm}[t]
\SetAlgoLined
\For{$t=0,\ldots,\nt-1$}{
\begin{align}
\nonumber
   \text{Take } u_t & =-(R+B^{\top}PB)^{-1}B^{\top}\left(PAx_t+\lambda\sum_{\tau=t}^{\nt-1}\left(F^\top\right)^{\tau-t} P\widehat{w}_{\tau}\right)\\
    \nonumber
   \text{Update } x_{t+1} &=Ax_t+Bu_t+w_t
\end{align}
}
\caption{$\lambda$-Confident Control}
\label{alg:lambda-confident}
\end{algorithm}

Specifically, the policy presented in Algorithm~\ref{alg:lambda-confident} works as follows.  Given a \textit{trust parameter} $0\leq \lambda\leq 1$, it implements a linear combination of \eqref{eq:MPC_explicit} and \eqref{eq:online_explicit}. Intuitively, the selection of $\lambda$ allows a trade-off between consistency and robustness based on the extent to which the predictions are trusted.  
Our main result shows a competitive ratio bound that is consistent with this intuition. A proof is given in Appendix \ref{appendix: Competitive Analysis}.


\begin{theorem}
\label{thm:upper_bound}
Under our model assumptions, with a fixed trust parameter $\lambda>0$, 
the $\lambda$-confident control in Algorithm~\ref{alg:lambda-confident} has a worst-case competitive ratio of at most 
\begin{align}
\label{eq:competitive_bound_fix_lambda}
\mathsf{CR}(\varepsilon)\leq  1+ 2\|H\|\min\left\{\left(\frac{\lambda^2}{\mathsf{OPT}}\varepsilon+\frac{(1-\lambda)^2}{C}\right),\left(\frac{1}{C}+\frac{\lambda^2 }{\mathsf{OPT}}\overline{W}\right)\right\}
\end{align} where $H\coloneqq B(R+B^\top P B)^{-1} B^\top$,
$\mathsf{OPT}$ denotes the optimal cost, $C>0$ is a constant that depends on $A,B,Q,R$ and
\begin{align}
\label{eq:varepsilon}
\varepsilon\left(F,P,e_0,\ldots,e_{\nt-1}\right) & \coloneqq  \sum_{t=0}^{\nt-1}\left\|\sum_{\tau=t}^{\nt-1}\left(F^\top\right)^{\tau-t}P\left(w_{\tau}-\widehat{w}_\tau\right)\right\|^2,\\
\nonumber
\overline{W}\left(F,P,\widehat{w}_0,\ldots,\widehat{w}_{\nt-1}\right) & \coloneqq \sum_{t=0}^{\nt-1}\left\|\sum_{\tau=t}^{\nt-1}\left(F^\top\right)^{\tau-t}P \widehat{w}_\tau\right\|^2.
\end{align}
\end{theorem}
From this result, we see that $\lambda$-confident control is guaranteed to be ${\big(1+\|H\|\frac{(1-\lambda)^2}{C}\big)}$-consistent and $\smash{\big(1+\|H\|\big(\frac{1}{C}+\frac{\lambda^2 }{\mathsf{OPT}} \overline{W}\big)\big)}$-robust. This highlights a trade-off between consistency and robustness such that if a large $\lambda$ is used (i.e., predictions are trusted), then consistency decreases to $1$, while the robustness  increases unboundedly.  In contrast, when a small $\lambda$ is used (i.e., predictions are distrusted), the robustness of the policy converges to the optimal value, but the consistency does not improve on the robustness value. Due to the time-coupling structure in the control system, the mismatches $e_t=\widehat{w}_t-w_t$ at different times contribute unequally to the system. As a result, the prediction error $\varepsilon$ in~\eqref{eq:prediction_error} and~\eqref{eq:varepsilon} is defined as a weighed quadratic sum of $\left(e_0,\ldots,e_{\nt-1}\right)$. {Moreover, the term $\mathsf{OPT}$ in~\eqref{eq:competitive_bound_fix_lambda} is common in the robustness and consistency analysis of online algorithms, such as~\cite{purohit2018improving,lykouris2018competitive,bamas2020primal,antoniadis2020online}.}



\section{Self-Tuning $\lambda$-Confident Control}
\label{sec:self-tuning-control}

While the $\lambda$-confident control finds a balance between consistency and robustness, selecting the optimal $\lambda$ parameter requires exogenous knowledge of the quality of the predictions $\varepsilon$, which is often not possible.  For example, black-box AI tools typically do not allow uncertainty quantification. In this section, we develop a self-tuning $\lambda$-confident control approach that learns to tune $\lambda$ in an online manner. We provide an upper bound on the regret of the self-tuning $\lambda$-confident control, compared with using the best possible $\lambda$ in hindsight, and a competitive ratio for the complete self-tuning algorithm.  These results provide the first worst-case guarantees for the integration of black-box AI tools into linear quadratic control. 

Our policy is described in Algorithm~\ref{alg:self} and is a ``follow the leader'' approach~\cite{hannan20164,kalai2005efficient}. At each time $t=0,\ldots,\nt-1$, it selects a $\lambda_t$ in order to minimize the gap between $\mathsf{ALG}$ and $\mathsf{OPT}$ in the previous $t$ rounds and chooses an action using the trust parameter $\lambda_t$. Then the state $x_t$ is updated to $x_{t+1}$ using the linear system dynamic in~\eqref{eq:linear_dynamic} and this process repeats. Note that the denominator of $\lambda_t$ is zero if and only if $\eta\left(\widehat{w};s,t-1\right)=0$ for all $s$. To make $\lambda_t$ well-defined, we set $\lambda=1$ for this case.

The key to the algorithm is the update rule for $\lambda_t$. Given previously observed perturbations and predictions, the goal of the algorithm is to find a greedy $\lambda_t$ that minimizes the gap between the algorithmic and optimal costs so that $\lambda_t
\coloneqq \min_{\lambda}\sum_{s=0}^{t-1}\psi_s^\top H \psi_s$ where $\psi_s\coloneqq      \sum_{\tau=s}^{t-1}\left(F^\top\right)^{\tau-s} P \left(w_\tau -  \lambda\widehat{w}_\tau\right).$  This can be equivalently written as
\begin{align}
\label{eq:4.1body}
    \lambda_t
    = \argmin_{\lambda}\sum_{s=0}^{t-1}\left[\left(\sum_{\tau=s}^{t-1}\left(F^\top\right)^{\tau-s} P (w_\tau-\lambda \widehat{w}_\tau) \right)^\top H \left(\sum_{\tau=s}^{t-1}\left(F^\top\right)^{\tau-s} P (w_\tau-\lambda \widehat{w}_\tau) \right)\right],
\end{align}
which is a quadratic function of $\lambda$. Rearranging the terms in~\eqref{eq:4.1body} yields the choice of $\lambda_t$ in the self-tuning control scheme.


\begin{algorithm}[t]
\caption{Self-Tuning $\lambda$-Confident Control}
\For{$t=0,\ldots,\nt-1$}{

\If{$t=0$ or $t=1$}{Initialize and choose $\lambda_0$}

\Else{
Compute a trust parameter $\lambda_t$
\begin{align*}
  \lambda_t = & \frac{\sum_{s=0}^{t-1}\left(\eta(w;s,t-1)\right)^{\top} H \left(\eta(\widehat{w};s,t-1)\right)}{\sum_{s=0}^{t-1}\left(\eta(\widehat{w};s,t-1)\right)^{\top} H \left(\eta(\widehat{w};s,t-1)\right)} \\
  &\text{ where }   \eta(w;s,t)\coloneqq \sum_{\tau=s}^{t}\left(F^\top\right)^{\tau-s} P w_\tau
\end{align*}}

Generate an action $u_t$ using $\lambda_t$-confident control (Algorithm~\ref{alg:lambda-confident})

Update
$
    x_{t+1} = Ax_t + Bu_t + w_t
$

}

\label{alg:self}
\end{algorithm}


{Algorithm~\ref{alg:self} is efficient since, in each time step, updating the $\eta$ values only  requires adding one more term. This means that the total computational complexity of $\lambda_t$ is $O(\nt^2 n^{\alpha})$, where $\alpha<2.373$, which is polynomial in both the time horizon length $\nt$ and state dimension $n$. According to the expression of $\lambda_t$ in Algorithm~\ref{alg:self}, at each time $t$, the terms $\eta(w;s,t-2)$ and $\eta(\widehat{w};s,t-2)$ can be pre-computed for all $s=0,\ldots,t-1$. Therefore, the recursive formula $\eta(w;s,t)\coloneqq \sum_{\tau=s}^{t}\left(F^\top\right)^{\tau-s} P w_\tau = \eta(w;s,t-1)+\left(F^\top\right)^{t-s} P w_t$ implies the update rule of the terms $\{\eta(w;s,t-1):s=0,\ldots,t-1\}$ in the expression of $\lambda_t$. This gives that, at each time $t$, it takes no  more than $O(\nt n^{\alpha})$ steps to compute $\lambda_t$  where $\alpha<2.373$ and $O(n^{\alpha})$ is the computational complexity of matrix multiplication.}

\subsection{Convergence}

We now move to the analysis of Algorithm \ref{alg:self}.  First, we study the convergence of $\lambda_t$, which depends on the variation of the predictions $\widehat{\mathbf{w}}\coloneqq (\widehat{w}_0,\ldots,\widehat{w}_{\nt-1})$ and the true perturbations $\mathbf{w}\coloneqq(w_0,\ldots,w_{\nt-1})$, where we use a boldface letter to represent a sequence of vectors. Specifically, our results are in terms of the variation of the predictions and perturbations, which we define as follows. The \emph{self-variation} $\mu_{\mathsf{VAR}}(\mathbf{y})$ of a sequence $\mathbf{y}\coloneqq (y_0,\ldots,y_{\nt-1})$ is defined as 
\begin{align*}
    \mu_{\mathsf{VAR}}(\mathbf{y})\coloneqq \sum_{s=1}^{\nt-1}\max_{\tau=0,\ldots,s-1}\left\|y_{\tau} - y_{\tau+\nt-s}\right\|.
\end{align*} 

The goal of the self-tuning algorithm is to converge to the optimal trust parameter $\lambda^*$ for the problem instance. To specify this formally, let $\mathsf{ALG}(\lambda_0,\ldots,\lambda_{\nt-1})$ be the algorithmic cost with adaptively chosen trust parameters $\lambda_0,\ldots,\lambda_{\nt-1}$ and denote by $\mathsf{ALG}(\lambda)$ the cost with a fixed trust parameter $\lambda$. Then, $\lambda^*$ is defined as
$    \lambda^*\coloneqq \min_{\lambda\in\mathbb{R}} \mathsf{ALG}(\lambda).
$ Further, let $W(t)\coloneqq \sum_{s=0}^{t}\eta(\widehat{w};s,t)^\top H {\eta}(\widehat{w};s,t)$.

We can now state a bound on the convergence rate of $\lambda_t$ to $\lambda^*$ under Algorithm \ref{alg:self}. The bound highlights that if the variation of the system perturbations and predictions is small, then the trust parameter $\lambda_t$ converges quickly to $\lambda^*$.  A proof can be found in Appendix~\ref{appendix: proof_of_lemma:lambda}.

\begin{lemma}
\label{lemma:lambda}
Assume $W(\nt)=\Omega(\nt)$ and $\lambda_t\in [0,1]$ for all $t=0,\ldots,\nt-1$. Under our model assumptions, the adaptively chosen trust parameters $\left(\lambda_0,\ldots,\lambda_{\nt}\right)$ by self-tuning control satisfy that for any $1< t\leq \nt$,
\begin{align*}
    \left|\lambda_t - \lambda^* \right| = O\left(\left(\mu_{\mathsf{VAR}}(\mathbf{w})+\mu_{\mathsf{VAR}}(\mathbf{\widehat{w}})\right)/{t}
\right).
\end{align*}
\end{lemma}

\subsection{Regret and Competitiveness}

Building on the convergence analysis, we now prove bounds on the regret and competitive ratio of Algorithm \ref{alg:self}.  These are the main results of the paper and represent the performance of an algorithm that adaptively determines the optimal trade-off between robustness and consistency.  

\paragraph{Regret.} We first study the regret as compared with the best, fixed trust parameter in hindsight, i.e., $\lambda^*$, whose corresponding worst-case competitive ratio satisfies the upper bound given in Theorem~\ref{thm:upper_bound}. 

Denote by $\mathsf{Regret}\coloneqq\mathsf{ALG}(\lambda_0,\ldots,\lambda_{\nt-1})-\mathsf{ALG}(\lambda^*)$ the \textit{regret} we consider where $(\lambda_0,\ldots,\lambda_{\nt-1})$ are the trust parameters selected by the self-tuning control scheme. Our main result is the following variation-based regret bound, which is proven in Appendix~\ref{appendix: proof_of_lemma:regret}. Define the denominator of $\lambda_t$ in Algorithm~\ref{alg:lambda-confident} by $W(t)\coloneqq \sum_{s=0}^{t}\eta(\widehat{w};s,t)^\top H {\eta}(\widehat{w};s,t)$.

\begin{lemma}
\label{lemma:regret}
Assume $W(t)=\Omega(\nt)$ and $\lambda_t\in [0,1]$ for all $t=0,\ldots,\nt-1$.
Under our model assumptions, for any $\mathbf{w}$ and $\mathbf{\widehat{w}}$, the regret of Algorithm \ref{alg:self} is bounded by
$$
        \mathsf{Regret} =  O\Big(\big(\mu_{\mathsf{VAR}}(\mathbf{w})+\mu_{\mathsf{VAR}}({\mathbf{\widehat{w}}})\big)\log\nt\Big).
$$
\end{lemma}

Note that the baseline we evaluate against in $\mathsf{Regret}=\mathsf{ALG}(\lambda_0,\ldots,\lambda_{\nt-1})-\mathsf{ALG}(\lambda^*)$ is stronger than baselines in previous \textit{static regret} analysis for LQR, such as~\cite{agarwal2019online,cohen2018online} where online controllers are compared with a linear control policy $u_t=-Kx_t$ with a strongly stable $K$. The baseline policy considered in our regret analysis is the $\lambda$-confident scheme (Algorithm~\ref{alg:lambda-confident}) with
\begin{align*}
    u_t =& -(R+B^{\top}PB)^{-1}B^{\top}\left(PAx_t+\lambda\sum_{\tau=t}^{\nt-1}\left(F^\top\right)^{\tau-t} P\widehat{w}_{\tau}\right) \\
    =& -Kx_t - \lambda(R+B^{\top}PB)^{-1}B^{\top}\sum_{\tau=t}^{\nt-1}\left(F^\top\right)^{\tau-t} P\widehat{w}_{\tau},
\end{align*}
which contains the class of strongly stable linear controllers as a special case.  Moreover, the regret bound in Lemma~\ref{lemma:regret} holds for any predictions $\widehat{w}_0,\ldots,\widehat{w}_{\nt-1}$. Taking $\widehat{w}_t=w_t$ for all $t=0,\ldots,\nt-1$, our regret directly compares $\mathsf{ALG}(\lambda_0,\ldots,\lambda_{\nt-1})$ with the optimal cost $\mathsf{OPT}$, and therefore, our regret also involves the \textit{dynamic regret} considered in~\cite{li2019online,yu2020power,zhang2021regret} for LQR as a special case.

To interpret this lemma, suppose the sequences of perturbations and predictions satisfy: 
\begin{align*}
\|\widehat{w}_{\tau} - \widehat{w}_{\tau+\nt-s}\| & \leq \rho(s),\\
\|w_{\tau} - w_{\tau+\nt-s}\| & \leq
\rho(s), \text{ for any }  s\geq 0, 0\leq \tau\leq s.
\end{align*}
These bounds correspond to an assumption of smooth variation in the disturbances and the predictions.  Note that it is natural for the disturbances to vary smoothly in applications such as tracking problems where the disturbances correspond to the trajectory and in such situations one would expect the predictions to also vary smoothly. For example, machine learning algorithms are often regularized to provide smooth predictions.

Given these smoothness bounds, we have that
\begin{align*}
    \mu_{\mathsf{VAR}}(\mathbf{w})+\mu_{\mathsf{VAR}}({\mathbf{\widehat{w}}})\leq \sum_{s=0}^{\nt-1}2\rho(s).
\end{align*}

Note that, as long as $\sum_{s=0}^{\nt-1}\rho(s)= o(\nt/\log\nt)$, the regret bound is sub-linear in $\nt$.
To understand how this bound may look in particular applications, suppose we have $\rho(s)=O(1/s)$.  In this case, regret is poly-logarithmic, i.e., $\mathsf{Regret}=O((\log\nt)^2)$.  If $\rho(s)$ is exponential the regret is even smaller, i.e., if $\rho(s) = O\left(r^s\right)$ for some $0<r<1$, then $\mathsf{Regret}=O(1)$.  The regret bound in Lemma~\ref{lemma:regret} depends on the variation of perturbations and predictions. Note that such a term commonly exists in regret analysis based on the ``follow the leader'' approach~\cite{hannan20164,kalai2005efficient}. For example, the regret analysis of the follow the optimal steady state (FOSS) method in~\cite{li2019online} contains a similar  ``path length'' term that captures the state variation and there is a fundamental limit on regret that depends on the variation budget (c.f. Theorem 3 in~\cite{li2019online}). There is a similar variation budget of the predictions or prediction errors in Theorem $1$ of~\cite{zhang2021regret}. In many robotics applications (e.g., the trajectory tracking and EV charging experiments in this paper shown in Section~\ref{sec:experiments}), each $w_t$ is from some desired smooth trajectory.

\paragraph{Competitive Ratio.} We are now ready to present our main result, which provides an upper bound on the competitive ratio of self-tuning control (Algorithm \ref{alg:self}). { Recall that, in Lemma~\ref{lemma:regret}, we bound the regret $\mathsf{Regret}\coloneqq\mathsf{ALG}(\lambda_0,\ldots,\lambda_{\nt-1})-\mathsf{ALG}(\lambda^*)$ and, in  Theorem~\ref{thm:upper_bound}, a competitive ratio bound is provided for the $\lambda$-confident control scheme, including $\mathsf{ALG}(\lambda^*)/\mathsf{OPT}$. Therefore, combining Lemma~\ref{lemma:regret} and Theorem~\ref{thm:upper_bound} leads to a novel competitive ratio bound for the self-tuning scheme (Algorithm~\ref{alg:self}). Note that compared with  Theorem~\ref{thm:upper_bound}, which also provides a competitive ratio bound for $\lambda$-confident control, Theorem~\ref{thm:competitive_self_tuning} below considers a competitive ratio bound for the self-tuning scheme in Algorithm~\ref{alg:self} where,  at each time $t$, a trust parameter $\lambda_t$ is determined by online learning and may be time-varying. }

\begin{theorem}
\label{thm:competitive_self_tuning}
Assume $W(\nt)=\Omega(\nt)$ and $\lambda_t\in [0,1]$ for all $t=0,\ldots,\nt-1$.
Under our model assumptions, the competitive ratio of Algorithm \ref{alg:self} is bounded by
\begin{align*}
\mathsf{CR}(\varepsilon)\leq 1+ 2\|H\|\frac{\varepsilon}{\mathsf{OPT}+C\varepsilon} + O\left(\frac{\left(\mu_{\mathsf{VAR}}(\mathbf{w})+\mu_{\mathsf{VAR}}({\mathbf{\widehat{w}}})\right)}{\mathsf{OPT}}\log\nt\right),
\end{align*} where $H$, $C$ $\mathsf{OPT}$ and $\varepsilon$ are defined in Theorem~\ref{thm:upper_bound}.
\end{theorem}

In contrast to the regret bound, Theorem~\ref{thm:competitive_self_tuning} states an upper bound on the competitive ratio $\mathsf{CR}(\varepsilon)$ defined in Section~\ref{sec:cr}, which indicates that $\mathsf{CR}(\varepsilon)$ scales as $1+{O(\varepsilon)}/\left({{\Theta(1)+\Theta(\varepsilon)}}\right)$ as a function of $\varepsilon$. As a comparison, the $\lambda$-confident control in Algorithm~\ref{alg:lambda-confident} has a competitive ratio upper bound that is linear in the prediction error $\varepsilon$ (Theorem~\ref{thm:upper_bound}). This improved dependency highlights the importance of learning the trust parameter adaptively.  

Our experimental results in the next section verify the implications of Theorem~\ref{thm:upper_bound} and Theorem~\ref{thm:competitive_self_tuning}. Specifically, the simulated competitive ratio of the self-tuning control (Algorithm~\ref{alg:self}) is a non-linear envelope of the simulated competitive ratios for $\lambda$-confident control with fixed trust parameters { and as a function of prediction error $\varepsilon$, it matches the implied competitive ratio upper bound $1+{O(\varepsilon)}/\left({{\Theta(1)+\Theta(\varepsilon)}}\right)$.}

Theorem~\ref{thm:competitive_self_tuning} is proven by combining Lemma~\ref{lemma:regret} with Theorem~\ref{thm:upper_bound}, which bounds the competitive ratios for fixed trust parameters.
\begin{proof}[Proof of Theorem~\ref{thm:competitive_self_tuning}]


Denote by $\mathsf{ALG}(\lambda_0,\ldots,\lambda_{\nt-1})$ the algorithmic cost of the self-tuning control scheme. We have
\begin{align}
\label{eq:4.30}
    \frac{\mathsf{ALG}(\lambda_0,\ldots,\lambda_{\nt-1})}{\mathsf{OPT}}\leq  \frac{\left|\mathsf{ALG}(\lambda_0,\ldots,\lambda_{\nt-1})-\mathsf{ALG}(\lambda^*)\right|}{\mathsf{OPT}} + \frac{\mathsf{ALG}(\lambda^*)}{\mathsf{OPT}}.
\end{align}

Using Theorem~\ref{thm:upper_bound},
\begin{align}
\nonumber
   \frac{\mathsf{ALG}(\lambda^*)}{\mathsf{OPT}}\leq & 1+2\|H\| \min\left\{\min_{\lambda}\left(\frac{\lambda^2}{\mathsf{OPT}}\varepsilon+\frac{(1-\lambda)^2}{C}\right),\min_{\lambda}\left(\frac{1}{C}+\frac{\lambda^2 }{\mathsf{OPT}}\overline{W}\right)\right\}\\
   \label{eq:4.31}
   = & 1+2\|H\|\min\left\{\frac{\varepsilon}{\mathsf{OPT}+\varepsilon C},\frac{1}{C}\right\} = 1+2\|H\|\frac{\varepsilon}{\mathsf{OPT}+\varepsilon C}.
\end{align}
Moreover, the regret bound in Lemma~\ref{lemma:regret} implies
\begin{align*}
  \frac{\left|\mathsf{ALG}(\lambda_0,\ldots,\lambda_{\nt-1})-\mathsf{ALG}(\lambda^*)\right|}{\mathsf{OPT}} =  O\left(\frac{\left(\mu_{\mathsf{VAR}}(\mathbf{w})+\mu_{\mathsf{VAR}}({\mathbf{\widehat{w}}})\right)\log\nt}{\mathsf{OPT}}\right),
\end{align*}
combing which with~\eqref{eq:4.31},~\eqref{eq:4.30} gives the results.
\end{proof}

\section{Case Studies}
\label{sec:experiments}

We now illustrate our main results using numerical examples and case studies to highlight the impact of the trust parameter $\lambda$ in $\lambda$-confident control and demonstrate the ability of the self-tuning control algorithm to learn the appropriate trust parameter $\lambda$. We consider three applications. The first is a robot tracking example where a robot is asked to follow locations of an unknown trajectory and the desired location is only revealed the time immediately before the robot makes a decision to modify its velocity. Predictions about the trajectory are available. However, the predictions can be untrustworthy so that they may contain large errors. The second is an adaptive battery-buffered electric vehicle (EV) charging problem where a battery-buffered charging station adaptively supplies energy demands of arriving EVs while maintaining the state of charge of the batteries as close to a nominal level as possible. Our third application considers a non-linear control problem--the Cart-Pole problem. Our $\lambda$-confident and self-tuning control schemes use a linearized model while the algorithms are tested with the non-linear environment. We use the third application to demonstrate the practicality of our algorithms by showing that they  do not only work for LQC problems, but also non-linear systems.

To illustrate the impact of randomness in prediction errors in our case studies,  the three applications all use different forms of random error models. For each selected distribution of $\mathbf{\widehat{w}}-\mathbf{w}$, we repeat the experiments multiple times and report the worst case with the highest algorithmic cost; see Appendix~\ref{appendix:experiments} for details.

\subsection{Application 1: Robot Tracking}

\textbf{Problem description.} The first example we consider is a two-dimensional robot tracking application~\cite{li2019online,yu2020competitive}. There is a robot controller following a fixed but unknown cloud-shaped trajectory (see Figures~\ref{fig:binomial_convergence_trajectory} and \ref{fig:gaussian_convergence_trajectory}), which is 
$$
    y_t\coloneqq \begin{bmatrix}
     2\cos(\pi t/30) + \cos(\pi t/5)\\
    2\sin(\pi t/30) + \sin(\pi t/5)
    \end{bmatrix}, \ \ t=0,\ldots,\nt-1.
$$
The robot controller's location at time $t+1$, denoted by $p_{t+1}\in\mathbb{R}^2$, depends on its previous location and its velocity $v_t\in\mathbb{R}^2$ such that $p_{t+1}=p_t+0.2v_t$ and at each time $t+1$, the controller is able to apply an adjustment $u_t$ to modify its velocity such that $v_{t+1}=v_t+0.2u_t$. Together, letting $x_t\coloneqq p_t-y_t$, this system can be recast in the canonical form in~\eqref{eq:linear_dynamic} as 
$$
    \begin{bmatrix}
    x_{t+1}\\
    v_{t+1}
    \end{bmatrix} = A    \begin{bmatrix}
    x_{t}\\
    v_{t}
    \end{bmatrix} + B u_t +w_t, \text{ with}
$$
$$
A\coloneqq
    \begin{bmatrix}
    1 & 0 & 0.2 & 0\\
    0 & 1 & 0 & 0.2\\
    0 & 0 & 1 & 0 \\
    0 & 0 & 0 & 1
    \end{bmatrix}, \ 
B\coloneqq
    \begin{bmatrix}
    0 & 0 \\
    0 & 0 \\
    0.2 & 0 \\
    0 & 0.2
    \end{bmatrix},
\text{ and } w_t \coloneqq Ay_t- y_{t+1}.
$$
To track the trajectory, the controller sets
$$
 Q\coloneqq
    \begin{bmatrix}
    1 & 0 & 0 & 0\\
    0 & 1 & 0 & 0\\
    0 & 0 & 0 & 0 \\
    0 & 0 & 0 & 0
    \end{bmatrix}   \text{ and } R\coloneqq
    \begin{bmatrix}
    10^{-2} & 0 \\
    0 & 10^{-2} 
    \end{bmatrix}.
$$
\begin{figure}[h!]
    \centering
    \begin{subfigure}[t]{1\textwidth}
    \centering
    \includegraphics[width=0.7\linewidth]{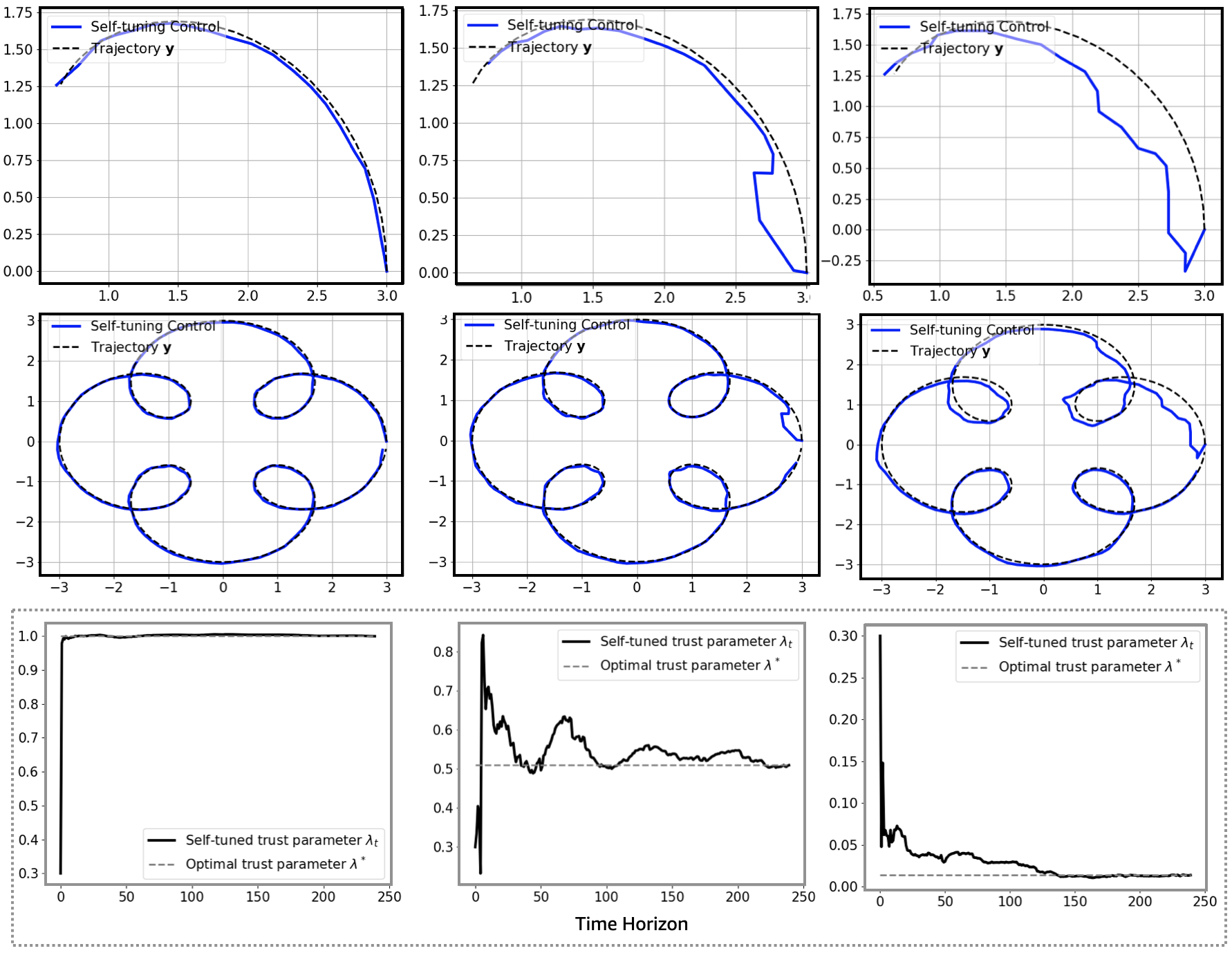}
          \caption{Left: low binomial prediction error with $c=10^{-2}$; middle: medium prediction error with $c= 10^{-1}$; right: high prediction error with $c=1$ where $c$ is a tuning parameter defined in Appendix~\ref{appendix:experiments}. }
\label{fig:binomial_convergence_trajectory}
\end{subfigure}
\\
\begin{subfigure}[t]{1\textwidth}
\includegraphics[width=0.7\linewidth]{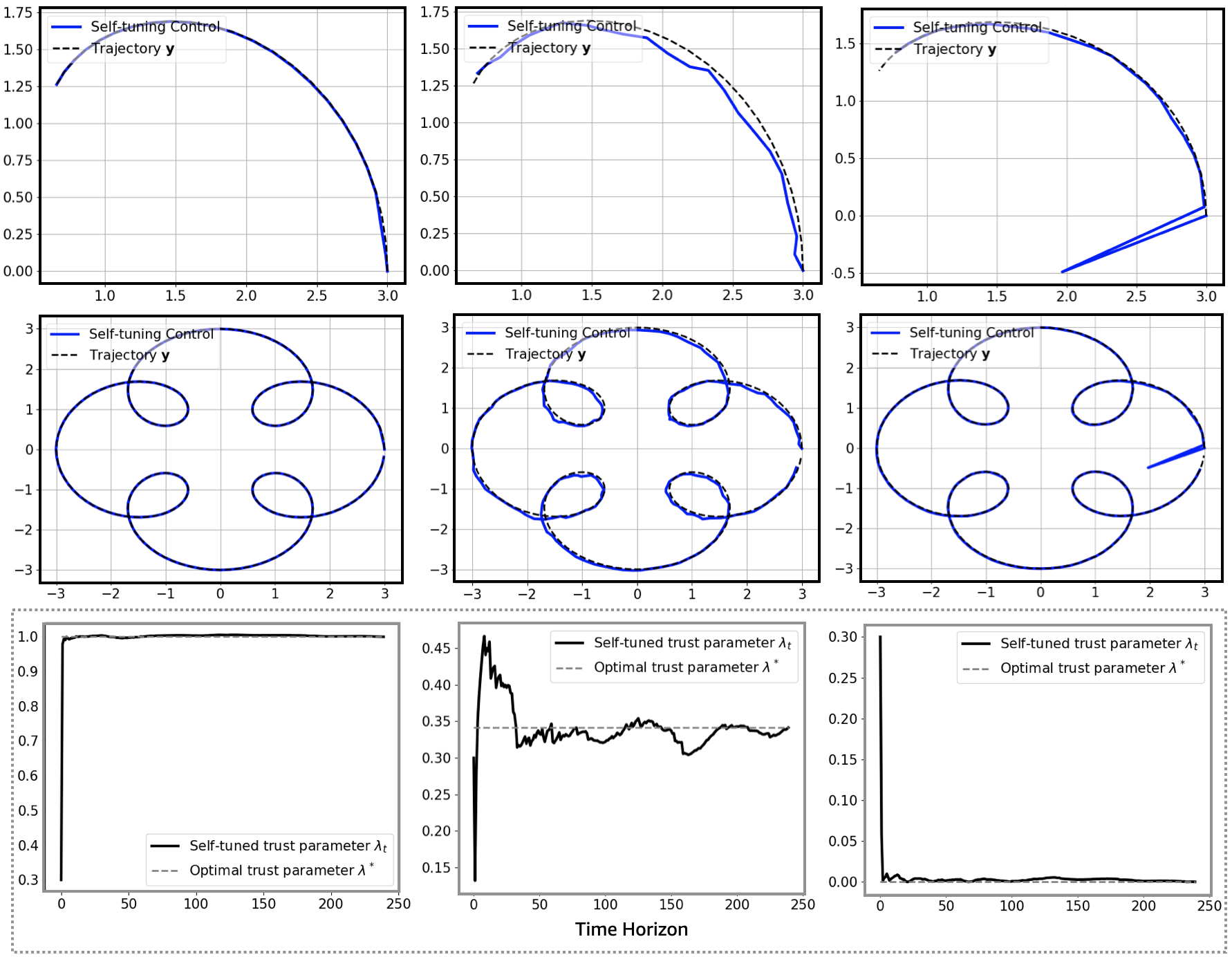}
\centering
          \caption{Left: low Gaussian prediction error with variance $\sigma^2=10^{-2}$; middle: medium Gaussian prediction error with variance $\sigma^2= 10^{-1}$; right: high Gaussian prediction error with variance $\sigma^2=1$. }
\label{fig:gaussian_convergence_trajectory}
\end{subfigure}
\caption{Tracking trajectories and trust parameters $(\lambda_0,\ldots,\lambda_{\nt-1})$ of the self-tuning control scheme. { The x-axis and y-axis in the top $6$ figures are locations of the robot. The y-axis in the bottom $3$ figures denotes the value of the trust parameter.}}
\end{figure}

\begin{figure}[h]
    \centering
    \includegraphics[scale=0.4]{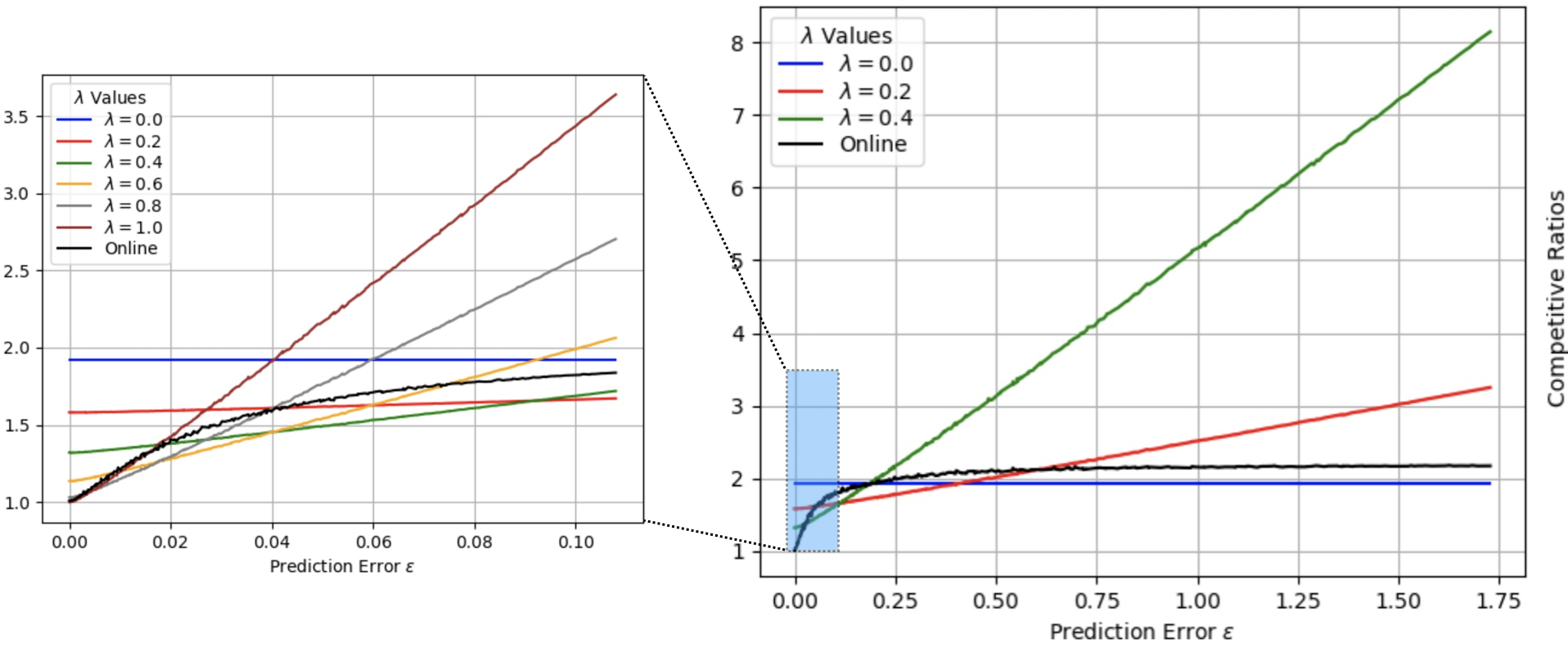}
          \caption{Impact of trust parameters and performance of self-tuning control for robot tracking.}
\label{fig:impact_tracking}
\end{figure}

\begin{figure}[h]
    \centering
    \includegraphics[scale=0.4]{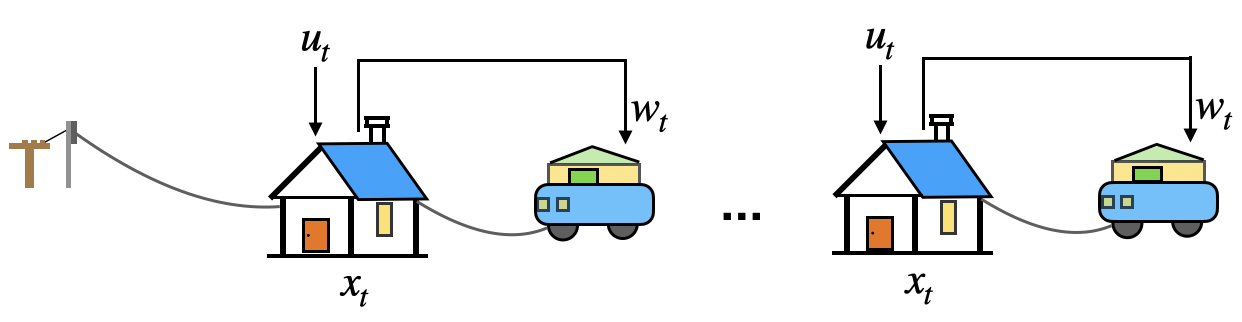}
          \caption{Adaptive battery-buffered EV charging modelled as a linear quadratic  control problem.}
\label{fig:impact_charging}
\end{figure}

\begin{figure}[h]
    \centering
    \includegraphics[scale=0.3]{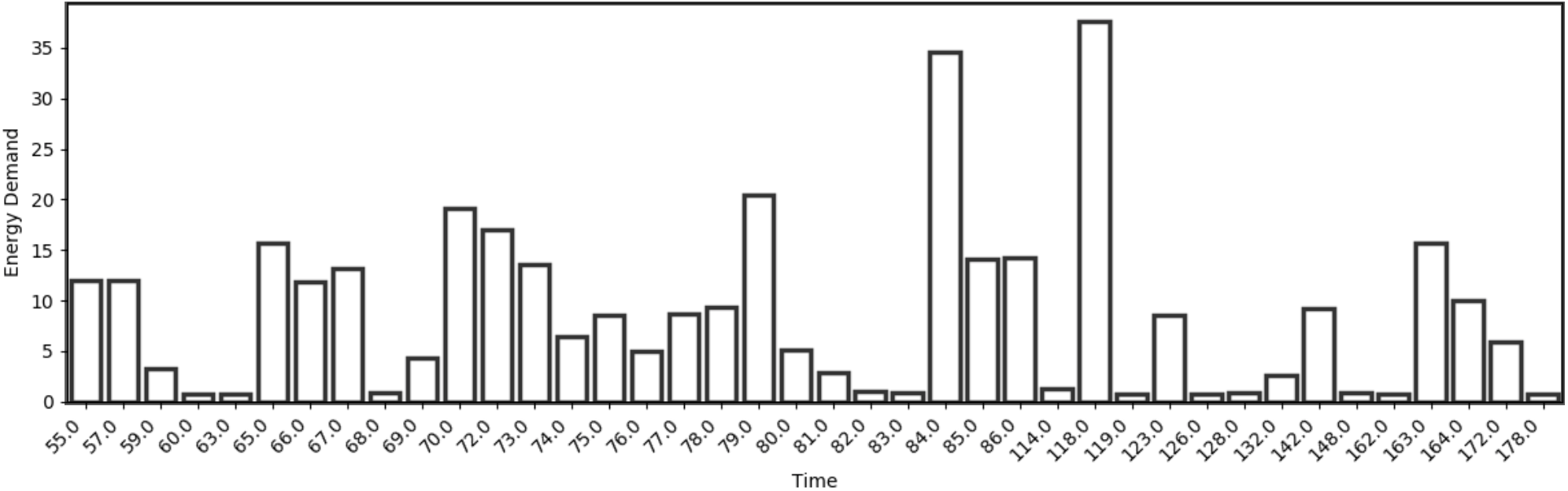}
          \caption{An example of the daily charging demands in ACN-Data~\cite{lee2019acn} on Nov 1st, $2018$.}
\label{fig:EV_data}
\end{figure}

\textbf{Experimental results.} In our first experiment, we demonstrate the convergence of the self-tuning scheme in Algorithm~\ref{alg:self}. To mimic the worst-case error, a random prediction error $e_t=\widehat{w}_t - w_t$ at each time $t$ is used. We then sample prediction error and implement our algorithm with several error instances and choose the one the worst competitive ratio. The details of settings can be found in Appendix~\ref{appendix:experiments}. To better simulate the task of tracking a trajectory and make it easier to observe the tracking accuracy, we ignore the cost of increasing velocity by setting $R$ as a zero matrix for Figure~\ref{fig:binomial_convergence_trajectory} and~Figure~\ref{fig:gaussian_convergence_trajectory}.

In Figure~\ref{fig:binomial_convergence_trajectory}, we observe that the tracking trajectory generated by the self-tuning scheme converges to the unknown trajectory $(y_1,\ldots,y_{\nt})$, regardless of the level of prediction error. 
We plot the tracking trajectories every $60$ time steps
with a scaling parameter (defined in Appendix~\ref{appendix:experiments}) $c=10^{-2}$ (left), $c=10^{-1}$ (mid) and $c=1$ (right), respectively. In all cases, we observe convergence of the trust parameters. Moreover, for a wide range of prediction error levels, without knowing the prediction error level in advance, the scheme is able to automatically switch its mode and become both consistent and robust by choosing an appropriate trust parameter $\lambda_t$ to accurately track the unknown trajectory. In Figure~\ref{fig:gaussian_convergence_trajectory}, we observe similar behavior when the prediction error is generated from Gaussian distributions.

Next, we demonstrate the performance of self-tuning control and the impact of trust parameters.
In Figure~\ref{fig:impact_tracking}, we depict the competitive ratios of the $\lambda$-confident control algorithm described in Section~\ref{sec:confident_control} with varying trust parameters, together with the competitive ratios of the self-tuning control scheme described in Algorithm~\ref{alg:self}. The label of the $x$-axis is the prediction error $\varepsilon$ (normalized by $10^3$), defined in~\eqref{eq:varepsilon}. We divide our results into two parts. The left sub-figure in Figure~\ref{fig:impact_tracking} considers a low-error regime where we observe that the competitive ratio of the self-tuning policy performs closely as the lower envelope formed by picking multiple trust parameters optimally offline. The right sub-figure in Figure~\ref{fig:impact_tracking} shows the performance of self-tuning for the case when the prediction error is high. For the high-error regime, the competitive ratio of the self-tuning control policy is close to those with the best fixed trust parameter. 


\subsection{Application 2: Adaptive Battery-Buffered EV Charging}

\textbf{Problem description.} We consider an adaptive battery-buffered Electric Vehicle (EV) charging problem. There is a charging station with $N$ chargers, with each charger connected to a battery energy storage system. Let $x_t$ be a vector in $\mathbb{R}_+^{N}$, whose entries represent the State of Charge (SoC) of the batteries at time $t$. The charging controller decides a charging schedule $u_t$ in $\mathbb{R}_+^{N}$ where each entry in $u_t$ is the energy to be charged to the $i$-th battery from external power supply at time $t$. The canonical form of the system can be represented by
\begin{align*}
    x_{t+1} =  A x_t + B u_t - w_t,
\end{align*}
where $A$ is an $N\times N$ matrix denotes the \textit{degradation} of battery charging levels and $B$ is an $N\times N$  diagonal matrix whose diagonal entry $0\leq B_i\leq 1$ represents the charging efficiency coefficient. In our experiments, without loss of generality, we assume $A$ and $B$ are identity matrices. The perturbation $w_t$ is defined as a length-$N$ vector, whose entry $w_t(i)=E$ when at time $t$ an EV arrives at charger $i$ and demands energy $E>0$; otherwise $w_t(i)=0$. Therefore the perturbations $(w_0,\ldots,w_{\nt-1})$ depend on the arrival of EVs and their energy demands. The  charging controller can only make a charging decision $u_t$ at time $t$ before knowing $w_t$ (as well as $w_{t+1},\ldots,w_{\nt-1}$) and the EVs that arrive at time $t$ (as well as future EV arrivals). The  goal of the adaptive battery-buffered EV charging problem is to maintain the battery SoC as close to a nominal value $\overline{x}$ as possible. Therefore, the charging controller would like to minimize $\sum_{t=0}^{\nt-1}\left(x_t-\overline{x}\right)^{\top} Q \left(x_t-\overline{x}\right) + u_t^{\top} R u_t$, equivalently, $\sum_{t=0}^{\nt-1}x_t^{\top} Q x_t + u_t^{\top} R u_t$ where $Q$ can be some positive-definite matrix and $R$ encodes the costs of external power supply. In our experiments, we set $Q$ as an identity matrix and $R=0.1\times Q$.

\textbf{Experimental results.} We show the performance of self-tuning control and the impact of trust parameters for adaptive EV charging in Figure~\ref{fig:impact_charging} and~\ref{fig:impact_charging_real}. In Figure~\ref{fig:impact_charging} we consider a synthetic case when EVs with $5$ kWh battery capacity arrive at a constant rate $0.2$, e.g., $1$ EV arrives every $5$ time slots. The results are divided into two parts. In Figure~\ref{fig:impact_charging_real}, we use daily data (ACN-Data) that contain EVs' energy demands, arrival times and departure times collected from a real-world adaptive EV charging network~\cite{lee2019acn}. We select a daily charging record on on Nov 1st, $2018$, depicted in Figure~\ref{fig:EV_data}. The left sub-figure considers a magnified low-error regime and the right sub-figure shows the performance of self-tuning for the case when the prediction error is high. For both regimes, the competitive ratios of the self-tuning control policy perform nearly as well as the lower envelope formed by picking multiple trust parameters optimally offline. We see in both Figure~\ref{fig:impact_tracking}  and Figure~\ref{fig:impact_charging} that with fixed trust parameters the competitive ratio is linear in $\varepsilon$, matching what Theorem~\ref{thm:upper_bound} indicates (in the sense of order in $\varepsilon$). Moreover, for the self-tuning scheme, in both Figure~\ref{fig:impact_tracking}  and Figure~\ref{fig:impact_charging}, we observe a competitive ratio $1+{O(\varepsilon)}/\left({\Theta(1)+\Theta(\varepsilon)}\right)$, which matches the competitive ratio bound given in Theorem~\ref{thm:competitive_self_tuning} in order sense (in $\varepsilon$).

\begin{figure}[h!]
    \centering
\begin{subfigure}[t]{1\textwidth}
    \centering
    \includegraphics[scale=0.405]{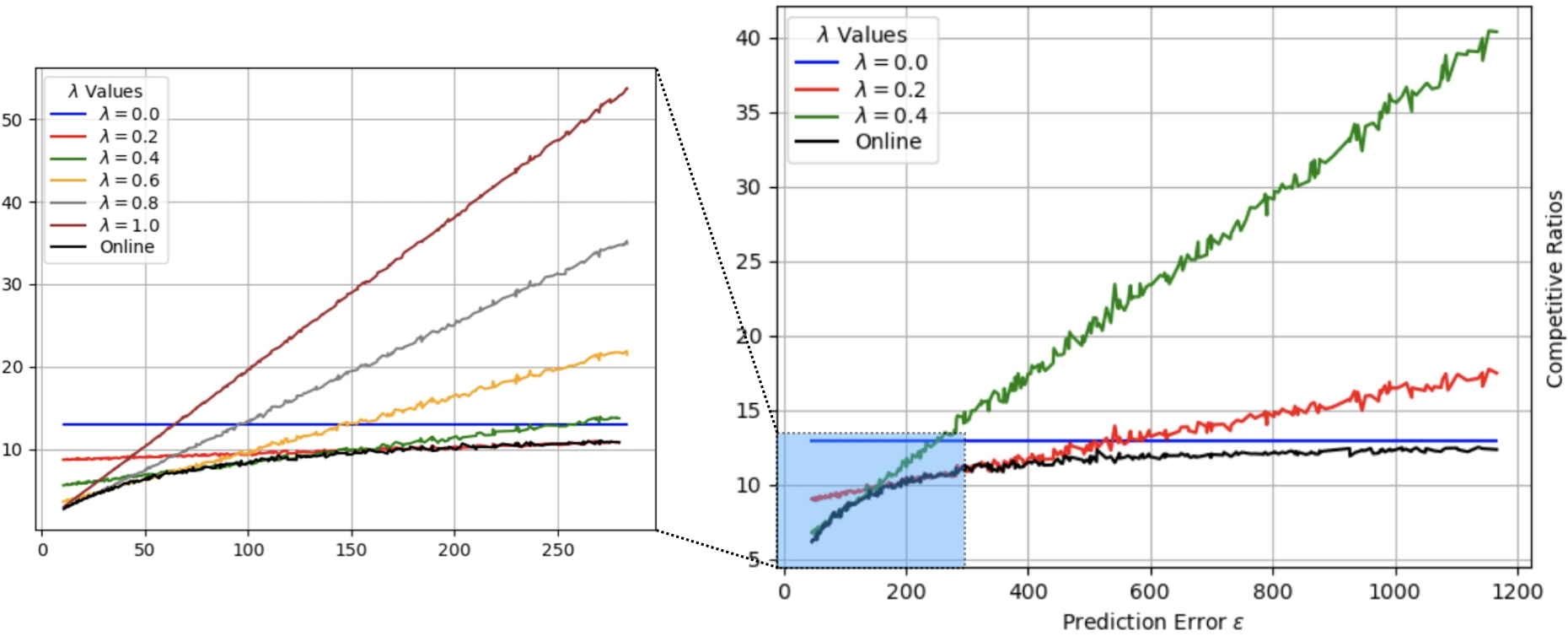}
          \caption{Experiments with  Synthetic EV charging.}
\label{fig:impact_charging}
\end{subfigure}
\\
\begin{subfigure}[t]{1\textwidth}
    \centering
    \includegraphics[scale=0.4]{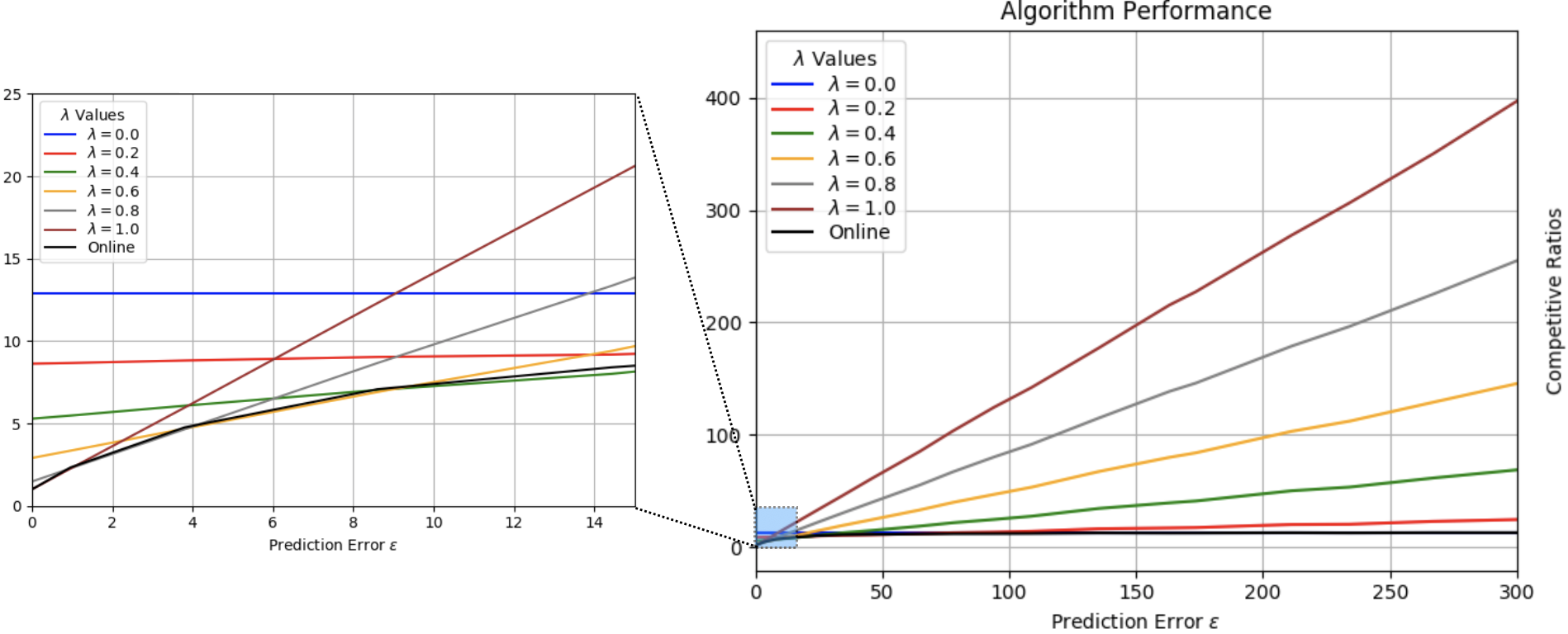}
          \caption{Experiments with daily EV charging data~\cite{lee2019acn}.}
\label{fig:impact_charging_real}
\end{subfigure}
\caption{Impact of trust parameters and performance of self-tuning control for adaptive battery-buffered EV charging with synthetic EV charging data (top) and realistic daily EV charging data~\cite{lee2019acn} (bottom).}
\end{figure}

\newpage

\subsection{Application 3: Cart-Pole}

\begin{figure}
    \begin{center}
    \includegraphics[width=0.25\textwidth]{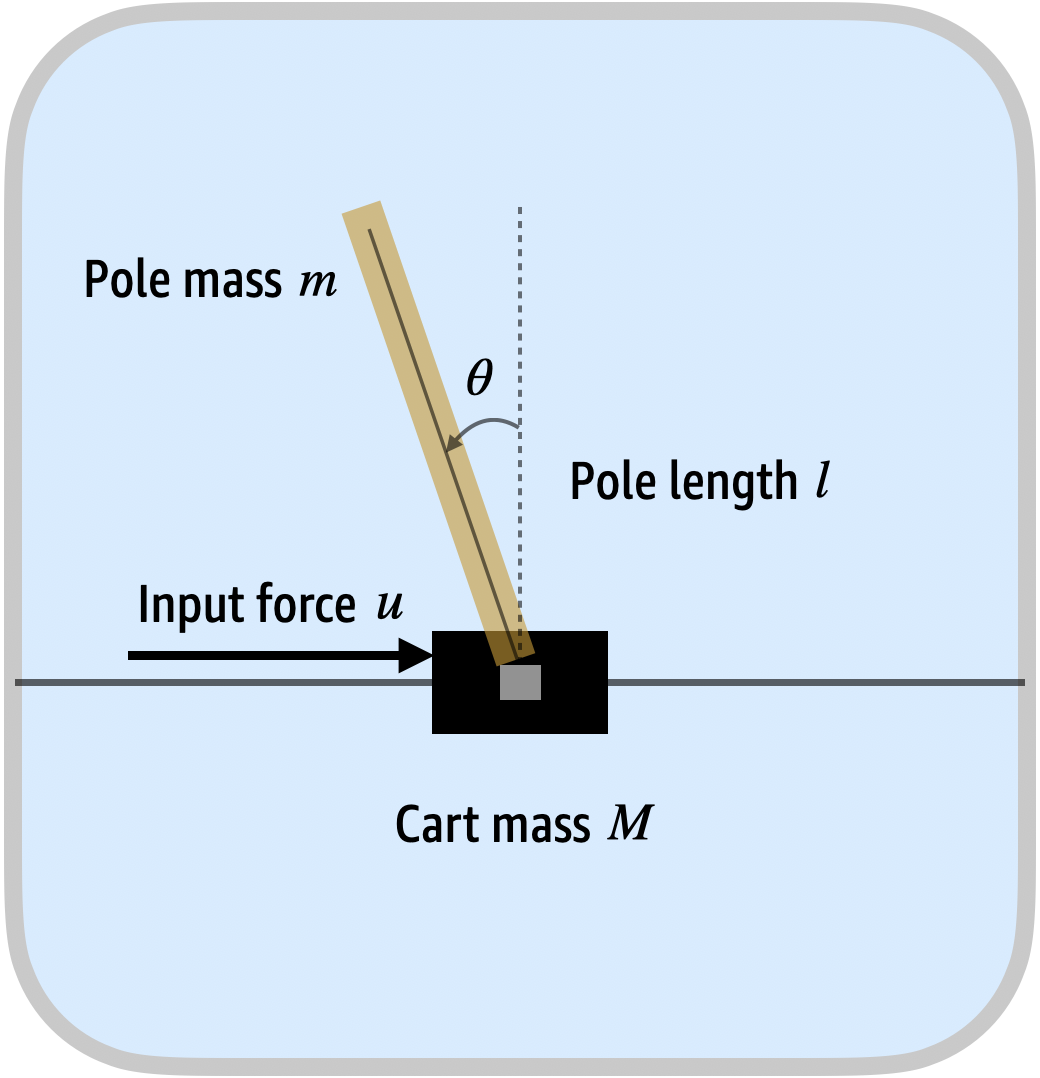}
    \end{center}
          \caption{\small The Cart-Pole model in Application 3.}
\label{fig:cartpole}
\end{figure}

\textbf{Problem description.}
The third set of experiments we consider is the classic Cart-Pole problem illustrated in Figure~\ref{fig:cartpole}. The goal of a controller is to stabilize the pole in the upright position. 
This is a widely studied nonlinear system. Neglecting friction, the dynamical equations of the Cart-Pole problem are
\begin{align}
\label{eq:cart0}
    \ddot{\theta} &= \frac{g \sin \theta + \cos \theta\left(\frac{-u-m l \dot{\theta}^2 \sin \theta}{m+M}\right)}{l\left(\frac{4}{3}-\frac{m \cos^2\theta}{m+M}\right)},\\
    \label{eq:cart1}
    \ddot{y} &= \frac{u+ m l\left(\dot{\theta}^2 \sin\theta - \ddot{\theta}\cos \theta\right)}{m+M}
\end{align}
where $u$ is the input force; $\theta$ is the angle between the pole and the vertical line; $y$ is the location of the pole; $g$ is the gravitational acceleration; $l$ is the pole length; $m$ is the pole mass; and $M$ is the cart mass. 
Taking $\sin\theta\approx \theta$ and $\cos\theta \approx 1$ and ignoring higher order terms, the dynamics of the Cart-Pole problem can be linearized as
{
\begin{align*} \frac{d}{dt}\begin{bmatrix}
    \dot{y}\\ \ddot{y} \\ \dot{\theta} \\ \ddot{\theta}
    \end{bmatrix} = \begin{bmatrix}
    0 & 1 & 0 & 0 \\ 0 & 0 & -\frac{m l g}{\eta(m +M)} & 0 \\ 0 & 0 & 0 & 1 \\ 0 & 0 & \frac{g}{\eta} & 0
    \end{bmatrix}\frac{d}{dt}\begin{bmatrix}
    {y}\\ \dot{y} \\ {\theta} \\ \dot{\theta}
    \end{bmatrix} + \begin{bmatrix}
   0\\ \frac{(m+M)\eta+m l}{(m+M)^2\eta}\\ 0 \\ -\frac{1}{(m+M)\eta}
    \end{bmatrix} u_t + w_t,
\end{align*}}
{
\noindent where in the above $\eta\coloneqq l\left(\frac{4}{3}-\frac{m}{m+M}\right)$ and, in our experiments, we set the cart mass $M=10.0 kg$, pole mass $m=1.0 kg$, pole length $l=10.0 m$ and gravitational acceleration $g=9.8 m/ s^2$. We set $Q=I$ and $R=10^{-3}$ and each $w_t$ is a fixed external force defined as $$60\times\left[
   0, \frac{(m+M)\eta+m l}{(m+M)^2\eta},0, -\frac{1}{(m+M)\eta}
    \right]^{\top}.$$}

\textbf{Experimental results.} 
We show the performance of the self-tuning control (Algorithm~\ref{alg:self}) and the impact of trust parameters for the Cart-Pole problem in Figure~\ref{fig:impact_cartpole}, together with the $\lambda$-confident control scheme in Algorithm~\ref{alg:lambda-confident} for several fixed trust parameters $\lambda$. The algorithms are tested using the true nonlinear dynamical equations in~\eqref{eq:cart0}-\eqref{eq:cart1}.

\begin{figure}[h]
    \centering
    \includegraphics[scale=0.45]{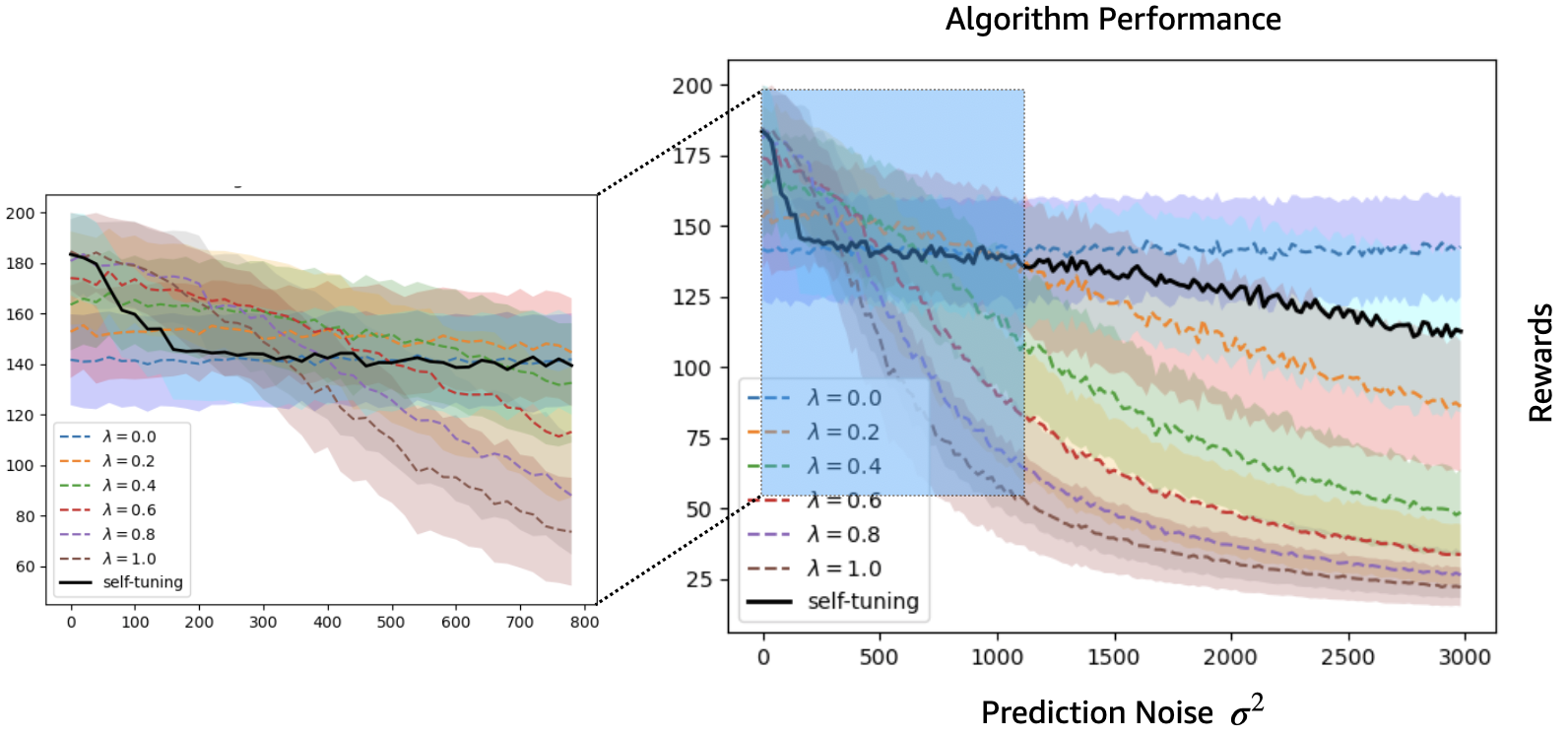}
          \caption{Impact of trust parameters and performance of self-tuning control for the Cart-Pole problem.}
\label{fig:impact_cartpole}
\end{figure}

In Figure~\ref{fig:impact_cartpole}, we change the variance $\sigma^2$ of the prediction noise $e_t=\widehat{w}_t - w_t$ at each time $t$ and plot the average episodic rewards in the OpenAI Gym environment~\cite{brockman2016openai}. Different from the worst-case settings in the previous two applications, we run episodes multiples times and show plot the mean rewards. The height of the shadow area in  Figure~\ref{fig:impact_cartpole} represents the standard deviation of the rewards. The detailed hyper-parameters are given in Section~\ref{appendix:experiments}. 
 Our results show that, despite the fact that the problem is \emph{nonlinear}, the self-tuning control algorithm using a linearized model is still able to automatically adjust the trust parameter $\lambda_t$ and achieves both consistency and robustness, regardless of the prediction error. In particular, it is close to the best algorithms for small prediction error while also staying among the best when prediction error is large.

\section{Concluding Remarks}

\label{sec:conclusion}
In this paper, we detail an approach that allows the use of black-box AI tools in a way that ensures worst-case performance bounds for linear quadratic control. Further, we demonstrate the effectiveness of our approach in multiple applications.   


There are many potential future directions that build on this work.
First, we consider a linear quadratic control problem in this paper, and an important extension will be to analyze the robustness and consistency of non-linear control systems. Second, our regret bound (Lemma~\ref{lemma:regret}) and competitive results (Theorem~\ref{thm:competitive_self_tuning}) are not tight when the variation of perturbations or predictions is high, therefore it is interesting to explore the idea of ``follow-the-regularized-leader''~\cite{shalev2011online,mcmahan2011follow} and understand if adding an extra regularizer in the update rule of $\lambda$ for self-tuning control can improve the convergence and/or the regret. Finally, characterizing a tight trade-off between robustness and consistency for linear quadratic control is of particular interest. For example, the results in~\cite{purohit2018improving,wei2020optimal} together imply a tight robustness and consistency trade-off for the ski-rental problem. It would be interesting to explore if it is possible to do the same for linear quadratic control.

\addcontentsline{toc}{section}{Bibliography}
\bibliographystyle{plain}
{\bibliography{main}}

\appendix

\section{Erratum}
In the published version of this paper~\cite{li2022robustness}, Lemma 5 states a bound that is not correct. Instead, the correct bound is:
\begin{equation*}
    \mathsf{Regret} \leq \|H\| \sum_{t=0}^{\nt-1} \left| \lambda_t - \lambda_{\nt} \right| \left| \lambda_t + \lambda_{\nt} \right| \left\| \sum_{\tau=t}^{\nt-1} \left( F^\top \right)^{\tau - t} P \widehat{w}_\tau \right\|,
\end{equation*}
as a result of Lemma~\ref{lemma:alg_opt} and the symmetry of $H$.

This yields a competitive ratio bound 
\begin{align*}
    \mathsf{CR}(\varepsilon)\leq 1+ 2\|H\|\frac{\varepsilon}{\mathsf{OPT}+C\varepsilon} + O\left(\frac{\left(\mu_{\mathsf{VAR}}(\mathbf{w})+\mu_{\mathsf{VAR}}({\mathbf{\widehat{w}}})\right)}{\mathsf{OPT}}\log\nt\right)
\end{align*}
instead of
\begin{align*}
    \mathsf{CR}(\varepsilon)\leq 1+ 2\|H\|\frac{\varepsilon}{\mathsf{OPT}+C\varepsilon} + O\left(\frac{\left(\mu_{\mathsf{VAR}}(\mathbf{w})+\mu_{\mathsf{VAR}}({\mathbf{\widehat{w}}})\right)^2}{\mathsf{OPT}}\right).
\end{align*}

\section{Useful Lemmas}
\label{appendix:lemma}

Before proceeding to the proofs of our main results, we present some useful lemmas.
We first present a lemma below from~\cite{yu2020competitive} that characterizes the difference between the optimal and the algorithmic costs.

\begin{lemma}[Lemma~10 in~\cite{yu2020competitive}]
\label{lemma:alg_opt}
 For any $\psi_t\in\mathbb{R}^\nn$, if at each time $t=0,\ldots,\nt-1$,
\begin{align*}
u_t = -(R+B^{\top}PB)^{-1}B^{\top}\left(PAx_t+\sum_{\tau=t}^{\nt-1}\left(F^\top\right)^{\tau-t} Pw_{\tau} - \psi_t\right),
\end{align*}
then the gap between the optimal cost $\mathsf{OPT}$ and the algorithm cost $\mathsf{ALG}$ induced by selecting control actions $(u_1,\ldots,u_{\nt})$ equals to
\begin{align}
\label{eq:gap}
 \mathsf{ALG}-\mathsf{OPT}=\sum_{t=0}^{\nt-1}\psi^\top_t H \psi_t
\end{align}
where $H\coloneqq B(R+B^\top P B)^{-1} B^\top$ and $F\coloneqq A-HPA$.
\end{lemma}

The next lemma describes the form of the optimal trust parameter.

\begin{lemma}
\label{lemma:optimal_tuning}
The optimal trust parameter $\lambda^*$ that minimizes $\mathsf{ALG}(\lambda)-\mathsf{OPT}$ is $\lambda^*=\lambda_{\nt}$.
\end{lemma}

\begin{proof}[Proof of~Lemma~\ref{lemma:optimal_tuning}]
The optimal trust parameter $\lambda^*$ is
\begin{align}
\label{eq:4.1}
    \lambda^*:
    = &\min_{\lambda}\sum_{s=0}^{\nt-1}\left[\left(\sum_{\tau=s}^{t-1}\left(F^\top\right)^{\tau-s} P (w_\tau-\lambda \widehat{w}_\tau) \right)^\top H \left(\sum_{\tau=s}^{\nt-1}\left(F^\top\right)^{\tau-s} P (w_\tau-\lambda \widehat{w}_\tau) \right)\right],
\end{align}
implying that $\lambda^*=\lambda_{\nt}$.
\end{proof}

Next, we note that the static regret depends on the convergence of $\lambda_t$.
\begin{lemma}
\label{lemma:regret_bound}
The static regret satisfies 
\begin{align*}
   {\mathsf{ALG}((\lambda_0,\ldots,\lambda_{\nt-1})) - \mathsf{ALG}(\lambda_{\nt})\leq \|H\| \sum_{t=0}^{\nt-1} \left| \lambda_t - \lambda_{\nt} \right| \left| \lambda_t + \lambda_{\nt} \right| \left\| \sum_{\tau=t}^{\nt-1} \left( F^\top \right)^{\tau - t} P \widehat{w}_\tau \right\|}.
\end{align*}
\end{lemma}
\begin{proof}[Proof of Lemma~\ref{lemma:regret_bound}]

Let $\mathsf{ALG}((\lambda_0,\ldots,\lambda_{\nt-1}))$ and $\mathsf{ALG}(\lambda_{\nt})$ denote the corresponding algorithm costs for using trust parameters $(\lambda_0,\ldots,\lambda_{\nt-1})$ and a fixed optimal trust parameter $\lambda_{\nt}$ in hindsight correspondingly. Since $H$ is symmetric, applying Lemma~\ref{lemma:alg_opt} twice implies
\begin{align}
\label{eq:4.7}
\mathsf{ALG}((\lambda_0,\ldots,\lambda_{\nt-1})) - \mathsf{ALG}(\lambda_{\nt})
 &\leq \|H\| \sum_{t=0}^{\nt-1} \left| \lambda_t - \lambda_{\nt} \right| \left| \lambda_t + \lambda_{\nt} \right| \left\| \sum_{\tau=t}^{\nt-1} \left( F^\top \right)^{\tau - t} P \widehat{w}_\tau \right\|\\
 &\leq  2\|H\| \sum_{t=0}^{\nt-1} \left| \lambda_t - \lambda_{\nt} \right| \left\| \sum_{\tau=t}^{\nt-1} \left( F^\top \right)^{\tau - t} P \widehat{w}_\tau \right\|,
\end{align}
where the last inequality follows from the fact that for any $t=1,\ldots,\nt$, $|\lambda_t| \leq 1$.

\end{proof}

\begin{lemma}
\label{lemma:series_1}
Suppose two real sequences $(V_1,\ldots,V_{\nt})$ and $(W_1,\ldots,W_{\nt})$ with $W_t> 0$ for all $1\leq t\leq \nt$, converge to $V_{\nt}$ and $W_{\nt}> 0$ such that for any integer $1\leq t\leq \nt$, $|V_t-V_{\nt}|\leq{C_1}/{t}$ and $|W_t-W_{\nt}|\leq{C_2}/{t}$ for some constants $C_1,C_2> 0$. Then the sequence $\left(\frac{V_1}{W_1},\ldots,\frac{V_{\nt}}{W_{\nt}}\right)$ converges to $\frac{V_{\nt}}{W_{\nt}}$ such that for any $1\leq t\leq \nt$,
\begin{align*}
    \left|\frac{V_t}{W_t}-\frac{V_{\nt}}{W_{\nt}}\right|\leq\frac{1}{t}\left(\frac{C_1\alpha_t+C_2}{|W_{\nt}|}\right)
\end{align*}
where $\alpha_t\coloneqq \max\{V_t/W_t\}$.
\end{lemma}
\begin{proof}[Proof of Lemma~\ref{lemma:series_1}]
Based on the assumption, for any $1\leq t\leq \nt$, we have that
\begin{align}
    \nonumber
    \left|\frac{V_t}{W_t}-\frac{V_{\nt}}{W_{\nt}}\right|=\left|\frac{V_t W_{\nt}-V_{\nt} W_t}{W_t W_{\nt}}\right|&=\left|\frac{V_t W_{\nt}-V_{t}W_{t}+V_{t}W_{t} -V_{\nt} W_t}{W_t W_{\nt}}\right|\\
    \nonumber
    &\leq \left|\frac{V_{t}\left(W_\nt -W_t\right)}{W_t W_{\nt}}\right| + \left|\frac{W_{t}\left(V_{\nt}- V_t\right)}{W_t W_{\nt}}\right|\\
    \nonumber
    &\leq \frac{1}{t}\left(\frac{C_1 |V_t|}{|W_t W_{\nt}|}+\frac{C_2}{\left|W_{\nt}\right|}\right).
\end{align}
Since $W_t\neq 0$ for all $1\leq t\leq \nt$ and $W_{\nt}\neq 0$, the lemma follows.
\end{proof}

\begin{lemma}
\label{lemma:series_2}
Suppose a sequence $(A_0,\ldots,A_{\nt-1})$ satisfies that for any integer $0\leq s \leq \nt-1$, $|A_s-A_{\nt}|\leq \rho(s)$. Then, for any $0\leq s \leq \nt$, $\left|\frac{1}{t}\left(\sum_{s=0}^{t}A_s\right)-A_{\nt}\right|\leq \frac{1}{t}\sum_{s=0}^{\nt-1}\rho(s)$.
\end{lemma}

\begin{proof}[Proof of Lemma~\ref{lemma:series_2}]
Based on the assumption,
\begin{align*}
    \left|\frac{1}{t}\sum_{s=0}^{t}A_s - A_{\nt}\right| =   \frac{1}{t}\left|\sum_{s=0}^{t}\left(A_s-A_{\nt}\right)\right| \leq \frac{1}{t}\sum_{s=0}^{t}\left|A_s-A_t\right| \leq \frac{1}{t}\sum_{s=0}^{\nt-1}\rho(s).
\end{align*}
\end{proof}

\section{Competitive Analysis}\label{appendix: Competitive Analysis}

Throughout, for notational convenience, we write
\begin{align*}
    W(t)\coloneqq \sum_{s=0}^{t}\eta(\widehat{w};s,t)^\top H {\eta}(\widehat{w};s,t), \quad \text{and} \  \
    V(t)\coloneqq \sum_{s=0}^{t} \eta(w;s,t)^\top H \eta(\widehat{w};s,t)
\end{align*}
where
\begin{align*}
    \eta(w;s,t)\coloneqq \sum_{\tau=s}^{t}\left(F^\top\right)^{\tau-s} P w_\tau, \quad \text{and} \  \  \eta(\widehat{w};s,t)\coloneqq \sum_{\tau=s}^{t}\left(F^\top\right)^{\tau-s} P \widehat{w}_\tau.
\end{align*}

We first prove the following theorem.

\begin{theorem}
With a fixed trust parameter $\lambda>0$, 
the $\lambda$-confident control in Algorithm~\ref{alg:lambda-confident} has a worst-case competitive ratio of at most 
\begin{align*}
\mathsf{CR}(\varepsilon)\leq 1+ 2\|H\|\min\left\{\left(\frac{\lambda^2}{\mathsf{OPT}}\varepsilon+\frac{(1-\lambda)^2}{C}\right),\left(\frac{1}{C}+\frac{\lambda^2 }{\mathsf{OPT}}\overline{W}\right)\right\}
\end{align*} where $H\coloneqq B(R+B^\top P B)^{-1} B^\top$,
$\mathsf{OPT}$ denotes the optimal cost, $C>0$ is a constant that depends on $A,B,Q,R$ and
\begin{align}
\nonumber
\varepsilon\left(F,P,e_0,\ldots,e_{\nt-1}\right) & \coloneqq  \sum_{t=0}^{\nt-1}\left\|\sum_{\tau=t}^{\nt-1}\left(F^\top\right)^{\tau-t}P\left(w_{\tau}-\widehat{w}_\tau\right)\right\|^2,\\
\nonumber
\overline{W}\left(F,P,\widehat{w}_0,\ldots,\widehat{w}_{\nt-1}\right) & \coloneqq \sum_{t=0}^{\nt-1}\left\|\sum_{\tau=t}^{\nt-1}\left(F^\top\right)^{\tau-t}P \widehat{w}_\tau\right\|^2.
\end{align}
\end{theorem}

\subsection{Proof of Theorem~\ref{thm:upper_bound}}
\label{appendix: proof_of_thm:upper_bound}

Denote by $\mathsf{ALG}$ the cost induced by taking actions $(u_0,\ldots,u_{\nt-1})$ in~Algorithm~\ref{alg:lambda-confident} and $\mathsf{OPT}$ the optimal total cost. Note that we assume $\mathsf{OPT}>0$. Lemma~\ref{lemma:alg_opt} implies that
\begin{align}
\label{alg2:4}
\mathsf{ALG} - \mathsf{OPT} = \sum_{t=0}^{\nt-1} \left(\sum_{\tau=t}^{\nt-1}\left(F^\top\right)^{\tau-t}P\left(w_t-\lambda\widehat{w}_\tau \right)\right)^\top H \left(\sum_{\tau=t}^{\nt-1}\left(F^\top\right)^{\tau-t}P\left(w_t-\lambda\widehat{w}_\tau \right)\right).
\end{align}
Therefore, with a sequence of actions $(u_1,\ldots,u_{\nt})$ generated by the $\lambda$-confident control scheme,  \eqref{alg2:4} leads to
\begin{align}
\nonumber
\mathsf{ALG} - \mathsf{OPT}
 \leq & \|H\|\sum_{t=0}^{\nt-1}\left\|\sum_{\tau=t}^{\nt-1}\left(F^\top\right)^{\tau-t}Pw_\tau-\lambda\sum_{\tau=t}^{\nt-1}\left(F^\top\right)^{\tau-t}P\widehat{w}_\tau\right\|^2\\
 \nonumber
 =& \|H\|\sum_{t=0}^{\nt-1}\left\|\sum_{\tau=t}^{\nt-1}\left(F^\top\right)^{\tau-t}Pw_\tau-\lambda\sum_{\tau=t}^{\nt-1}\left(F^\top\right)^{\tau-t}P\left({w}_\tau+e_\tau\right)\right\|^2\\
  \nonumber
 =& \|H\|\sum_{t=0}^{\nt-1}\left\|(1-\lambda)\sum_{\tau=t}^{\nt-1}\left(F^\top\right)^{\tau-t}Pw_\tau-\lambda\sum_{\tau=t}^{\nt-1}\left(F^\top\right)^{\tau-t}Pe_\tau\right\|^2\\
 \nonumber
 \leq & 2\|H\|\Bigg((1-\lambda)^2\sum_{t=0}^{\nt-1}\left\|\sum_{\tau=t}^{\nt-1}\left(F^\top\right)^{\tau-t}Pw_\tau\right\|^2+ \lambda^2\sum_{t=0}^{\nt-1}\left\|\sum_{\tau=t}^{\nt-1}\left(F^\top\right)^{\tau-t}P e_\tau\right\|^2\Bigg)
 \end{align}
 where $e_t \coloneqq \widehat{w}_t -w_T$ for all $t=0,\ldots,\nt-1$.
Moreover, denoting by $x^*_t$ and $u_t^*$ the offline optimal state and action at time $t$,  the optimal cost satisfies
\begin{align}
\nonumber
  \mathsf{OPT}= & \sum_{t=0}^{\nt-1}(x_t^*)^{\top} Q x^*_t + (u_t^*)^{\top} R u^*_t + (x^*_{\nt})^{\top}P x^*_{\nt}\\
  \label{eq:4.60}
  \geq &\sum_{t=0}^{\nt-1} \lambda_{\min}(Q)\left\|x^*_t\right\|^2 + \lambda_{\min}(R)\|u^*_t\|^2 + \lambda_{\min}(P)\|x^*_{\nt}\|^2\\
  \nonumber
  \geq & 2D_0 \sum_{t=0}^{\nt-1} \left(\|Ax^*_t\|^2+\|B u^*_t\|^2\right)+\frac{1}{2}\sum_{t=0}^{\nt-1}\lambda_{\min}(Q)\|x^*_{t}\|^2+\lambda_{\min}(P)\|x^*_{\nt}\|^2\\
  \nonumber
    \geq & D_0 \sum_{t=0}^{\nt-1} \left\|Ax^*_t+Bu^*_t\right\|^2+\frac{1}{2}\sum_{t=0}^{\nt-1}\lambda_{\min}(Q)\|x^*_{t}\|^2+\lambda_{\min}(P)\|x^*_{\nt}\|^2\\
\nonumber
= & D_0\sum_{t=0}^{\nt-1}\left\|x^*_{t+1}-w_t\right\|^2+\frac{1}{2}\sum_{t=0}^{\nt-1}\lambda_{\min}(Q)\|x^*_{t}\|^2+ \lambda_{\min}(P)\|x^*_{\nt}\|^2\\
\label{eq:4.80}
\geq & \frac{D_0}{2}\sum_{t=0}^{\nt-1}\left\|w_t\right\|^2  +\left(\frac{\lambda_{\min}(Q)}{2}-D_0\right)\sum_{t=0}^{T-1}\|x^*_t\|^2 + \left(\lambda_{\min}(P)-C\right)\|x^*_{\nt}\|^2
\end{align}
for some constant $0<D_0<\min\{\lambda_{\min}(P),\lambda_{\min}(Q)/2\}$ that depends on $Q,R$ and $K$
where in~\eqref{eq:4.60}, $\lambda_{\min}(Q)$, $\lambda_{\min}(R)$  and $\lambda_{\min}(P)$ are the smallest eigenvalues of positive definite matrices $Q,R$ and $P$, respectively.
Let
$\psi_t\coloneqq \sum_{\tau=0}^{\nt-t-1}\left(F^\top\right)^\tau P w_{t+\tau}$. Note that $F=A-BK$ and we define $\rho\coloneqq \frac{1+\rho(F)}{2}<1$ where  $\rho(F)$ denotes the spectral radius of $F$. From Gelfand’s formula, there exists a constant $D_1\geq 0$ such that $\|F^t\|\leq D_1\rho^{t}$ for all $t\geq 0$.
Therefore,
\begin{align}
\nonumber
 \sum_{t=0}^{\nt-1}\left\|\psi_t\right\|^2 =&  \sum_{t=0}^{\nt-1}\left\|\sum_{\tau=0}^{\nt-t-1}\left(F^{\top}\right)^{\tau}P w_{t+\tau}\right\|^2\\
 \nonumber
 \leq &  D_1^2 \|P\|^2\sum_{t=0}^{\nt-1}\left(\sum_{\tau=0}^{\nt-t-1}\rho^\tau\|w_{t+\tau}\|\right)^2\\
 \nonumber
 = &  D_1^2 \|P\|^2\sum_{t=0}^{\nt-1}\sum_{\tau=0}^{\nt-t-1}\sum_{\ell=0}^{\nt-t-1}\rho^{\tau}\rho^{\ell}\|w_{t+\tau}\|\|w_{t+\ell}\|\\
 \label{eq:4.70}
 \leq & \frac{D_1^2}{2} \|P\|^2\sum_{t=0}^{\nt-1}\sum_{\tau=0}^{\nt-t-1}\sum_{\ell=0}^{\nt-t-1}\rho^{\tau}\rho^{\ell}\left(\|w_{t+\tau}\|^2+\|w_{t+\ell}\|^2\right).
\end{align}
Continuing from~\eqref{eq:4.70},
\begin{align}
\nonumber
 \sum_{t=0}^{\nt-1}\left\|\psi_t\right\|^2\leq &\frac{D_1^2}{2} \|P\|^2\left(\sum_{\ell=0}^{\nt-t-1}\rho^{\ell}\right)\sum_{t=0}^{\nt-1}\sum_{\tau=0}^{\nt-t-1}\rho^{\tau}\|w_{t+\tau}\|^2 \\
 \quad &+ \frac{D_1^2}{2} \|P\|^2\left(\sum_{\tau=0}^{\nt-t-1}\rho^{\tau}\right)\sum_{t=0}^{\nt-1}\sum_{\ell=0}^{\nt-t-1}\rho^{\ell}\|w_{t+\ell}\|^2\\
 \nonumber
 \leq & \frac{D_1^2}{1-\rho} \|P\|^2\sum_{t=0}^{\nt-1}\sum_{\tau=0}^{\nt-t-1}\rho^{\tau}\|w_{t+\tau}\|^2\\
 \nonumber
 \leq  & \frac{D_1^2}{1-\rho} \|P\|^2\sum_{t=0}^{\nt-1}\sum_{\tau=0}^{\nt-1}\rho^{\tau}\|w_{(t+\tau)\mod \nt}\|^2\\
\nonumber
 = & \frac{D_1^2}{1-\rho} \|P\|^2\left(\sum_{\tau=0}^{\nt-1}\rho^{\tau}\right)\left(\sum_{t=0}^{\nt-1}\|w_{t}\|^2\right)\\
  \label{eq:4.90}
 \leq & \frac{D_1^2}{(1-\rho)^2} \|P\|^2 \sum_{t=0}^{\nt-1}\|w_t\|^2.
\end{align}
Putting~\eqref{eq:4.90} into~\eqref{eq:4.80}, we obtain
\begin{align*}
  \mathsf{OPT}\geq & \frac{D_0 (1-\rho)^2 }{D_1^2\|P\|^2}  \sum_{t=0}^{\nt-1}\|\psi_t\|^2,
  \end{align*}
  which implies that
\begin{align*}
\frac{\mathsf{ALG} - \mathsf{OPT}}{\mathsf{OPT}}&\leq 2\|H\|\left(\frac{\lambda^2}{\mathsf{OPT}}\varepsilon+\frac{(1-\lambda)^2}{C}\right)
\end{align*}
where $C\coloneqq \frac{D_0 (1-\rho)^2 }{D_1^2\|P\|^2}$ and
\begin{align*}
    \varepsilon\coloneqq & \sum_{t=0}^{\nt-1}\left\|\sum_{\tau=t}^{\nt-1}\left(F^\top\right)^{\tau-t}P\left(w_{\tau}-\widehat{w}_\tau\right)\right\|^2.
\end{align*}

To obtain the second bound, noting that
\begin{align}
 \nonumber
\mathsf{ALG} - \mathsf{OPT}
 \leq & \|H\|\sum_{t=0}^{\nt-1}\left\|\sum_{\tau=t}^{\nt-1}\left(F^\top\right)^{\tau-t}Pw_\tau-\lambda\sum_{\tau=t}^{\nt-1}\left(F^\top\right)^{\tau-t}P\widehat{w}_\tau\right\|^2\\
\nonumber
 \leq & 2\|H\|\sum_{t=0}^{\nt-1}\Bigg(\left\|\sum_{\tau=t}^{\nt-1}\left(F^\top\right)^{\tau-t}P w_{\tau}\right\|^2 + \lambda^2 \left\|\sum_{\tau=t}^{\nt-1} \left(F^\top\right)^{\tau-t}P\widehat{w}_\tau\right\|^2\Bigg).
\end{align}
Noting that $W\coloneqq \sum_{t=0}^{\nt-1}\left\|\sum_{\tau=t}^{\nt-1}\left(F^\top\right)^{\tau-t}P \widehat{w}_\tau\right\|^2$,
therefore,
\begin{align*}
\frac{\mathsf{ALG} - \mathsf{OPT}}{\mathsf{OPT}}&\leq 2\|H\|\left(\frac{1}{C}+\frac{\lambda^2}{\mathsf{OPT}}W\right)
\end{align*}
for some constant $C>0$ that depends on $A,B,Q$ and $R$.

\section{Regret Analysis of Self-tuning Control}
Throughout, for notational convenience, we write
\begin{align*}
    W(t)\coloneqq \sum_{s=0}^{t}\eta(\widehat{w};s,t)^\top H {\eta}(\widehat{w};s,t), \quad \text{and} \  \
    V(t)\coloneqq \sum_{s=0}^{t} \eta(w;s,t)^\top H \eta(\widehat{w};s,t)
\end{align*}
where
\begin{align*}
    \eta(w;s,t)\coloneqq \sum_{\tau=s}^{t}\left(F^\top\right)^{\tau-s} P w_\tau, \quad \text{and} \  \  \eta(\widehat{w};s,t)\coloneqq \sum_{\tau=s}^{t}\left(F^\top\right)^{\tau-s} P \widehat{w}_\tau.
\end{align*}

\subsection{Proof of Lemma~\ref{lemma:lambda}}
\label{appendix: proof_of_lemma:lambda}

In this section, we show the proof of Lemma~\ref{lemma:regret} and Lemma~\ref{lemma:lambda}. 
We begin with rewriting $ \lambda_t - \lambda_{\nt}$ as below.
\begin{align}
    \label{eq:4.6}
    \lambda_t - \lambda_{\nt} = \frac{V(t-1)}{ W(t-1)} - \frac{V(\nt-1)}{ W(\nt-1)} = \frac{\frac{V(t-1)}{t-1}}{\frac{W(t-1)}{(t-1)}} -\frac{\frac{V(\nt-1)}{\nt-1}}{\frac{W(\nt-1)}{\nt-1}}.
\end{align}
Applying Lemma~\ref{lemma:series_1}, it suffices to prove that for any $1\leq t\leq \nt$, $\left|\frac{1}{\nt}V(\nt)-\frac{1}{t}V(t)\right|\leq \frac{C_1}{t}$ and  $\left|\frac{1}{\nt}W(\nt)-\frac{1}{t}W(t)\right|\leq  \frac{C_2}{t}$ for some constants $C_1>0$ and $C_2>0$. In the sequel, we show the bound on $\left|\frac{1}{\nt}V(\nt)-\frac{1}{t}V(t)\right|$ and the bound on $\left|\frac{1}{\nt}W(\nt)-\frac{1}{t}W(t)\right|$ follows using the same argument.
Continuing from~\eqref{eq:4.6},
\begin{align}
\nonumber
  \left|\frac{1}{\nt}V(\nt)-\frac{1}{t}V(t)\right| 
  & \leq \underbrace{\left|\frac{1}{\nt}\sum_{s=0}^{\nt} \eta(w;s,\nt)^\top H \eta(\widehat{w};s,\nt) -\frac{1}{t}\sum_{s=0}^{t} \eta(w;s,\nt)^\top H \eta(\widehat{w};s,\nt)\right|}_{=:(a)}\\
 \label{eq:centralbound}
  & \quad + \underbrace{\left|\frac{1}{t}\sum_{s=0}^{t} \eta(w;s,t)^\top H \eta(\widehat{w};s,t) -\frac{1}{t}\sum_{s=0}^{t} \eta(w;s,\nt)^\top H \eta(\widehat{w};s,\nt)\right|}_{=:(b)}.
\end{align}

In the following, we deal with the terms (a) and (b) separately.

\subsubsection{Upper bound on (a)}

To bound the term (a) in~\eqref{eq:centralbound}, we notice that (a) can be regarded as a difference between two algebraic means. 
Rewriting the first mean in (a), we get
\begin{align*}
\sum_{s=0}^{\nt} \eta(w;s,\nt)^\top H \eta(\widehat{w};s,\nt)
& = \sum_{s=0}^{\nt} \left(\sum_{\tau=s}^{\nt}\left(F^\top\right)^{\tau-s} P w_\tau\right)^\top H \left(\sum_{\tau=s}^{\nt}\left(F^\top\right)^{\tau-s} P \widehat{w}_\tau\right)\\
& = \sum_{s=0}^{\nt} \left(\sum_{\tau=0}^{\nt-s}\left(F^\top\right)^{\tau} P w_{\tau+s}\right)^\top H \left(\sum_{\tau=0}^{\nt-s}\left(F^\top\right)^{\tau} P \widehat{w}_{\tau+s}\right)\\
& = \sum_{s=0}^{\nt} \overline{\eta}(w;s,\nt)^\top H   \overline{\eta}(\widehat{w};s,\nt),
\end{align*}
where for notational convenience, for $s\leq \nt$ we have defined two series 
\begin{align*}
  \overline{\eta}(\widehat{w};s,\nt)\coloneqq \sum_{\tau=0}^{s}\left(F^\top\right)^{\tau} P \widehat{w}_{\tau+\nt-s}, \quad \text{and} \ \ 
\overline{\eta}(w;s,\nt)\coloneqq \sum_{\tau=0}^{s}\left(F^\top\right)^{\tau} P w_{\tau+\nt-s}.
\end{align*}
We state a lemma below, which states that the sequence $\left(  \overline{\eta}(\widehat{w};0,\nt),\ldots,  \overline{\eta}(\widehat{w};\nt,\nt)\right)$ satisfies the assumption in Lemma~\ref{lemma:series_2}.

\begin{lemma}
\label{lemma:bound_eta}
Given an integer $s$ with $0\leq s\leq\nt$, we have
\begin{align*}
&\left| \overline{\eta}(w;\nt,\nt)^\top H \overline{\eta}(\widehat{w};\nt,\nt) -  \overline{\eta}(w;s,\nt)^\top H \overline{\eta}(\widehat{w};s,\nt)\right| \\ \leq  &2\|H\|\left(\frac{C \|P\|}{1-\rho}\right)^2\left(2\rho^{s+1}\overline{w}\widehat{w} +\max_{\tau}\left\|\widehat{w}_{\tau} - \widehat{w}_{\tau+\nt-s}\right\| +\max_{\tau}\left\|w_{\tau} - w_{\tau+\nt-s}\right\|\right).
\end{align*}
\end{lemma}

\begin{proof}[Proof of Lemma~\ref{lemma:bound_eta}]
With $s\leq \nt$, according to the definitions of $\overline{\eta}(w;s,\nt)$, $\overline{\eta}(w;\nt,\nt)$, $\overline{\eta}(\widehat{w};s,\nt)$ and $\overline{\eta}(\widehat{w};\nt,\nt)$, we obtain
\begin{align*}
 \overline{\eta}(w;\nt,\nt) &= \overline{\eta}(w;s,\nt) +  \sum_{\tau=s+1}^{\nt}\left(F^\top\right)^{\tau} P w_{\tau} + \sum_{\tau=0}^{s}\left(F^\top\right)^{\tau} P\left(w_{\tau} - w_{\tau+\nt-s}\right), \\
 \overline{\eta}(\widehat{w};\nt,\nt) &= \overline{\eta}(\widehat{w};s,\nt) + \sum_{\tau=s+1}^{\nt}\left(F^\top\right)^{\tau} P \widehat{w}_{\tau} + \sum_{\tau=0}^{s}\left(F^\top\right)^{\tau} P\left(\widehat{w}_{\tau} - \widehat{w}_{\tau+\nt-s}\right),
\end{align*}
implying that
\begin{align}
\nonumber
  & \overline{\eta}(w;\nt,\nt)^\top H \overline{\eta}(\widehat{w};\nt,\nt) -  \overline{\eta}(w;s,\nt)^\top H \overline{\eta}(\widehat{w};s,\nt)\\
 \label{eq:4.12}
  = &
 \overline{\eta}(w;s,\nt)^\top  H \xi_2+ \xi_1^{\top} H \overline{\eta}(\widehat{w};s,\nt) + \xi_1^{\top} H \xi_2
\end{align}
where
\begin{align*}
    \xi_1\coloneqq  & \sum_{\tau=s+1}^{\nt}\left(F^\top\right)^{\tau} P w_{\tau} + \sum_{\tau=0}^{s}\left(F^\top\right)^{\tau} P\left(w_{\tau} - w_{\tau+\nt-s}\right),\\
    \xi_2\coloneqq  & \sum_{\tau=s+1}^{\nt}\left(F^\top\right)^{\tau} P \widehat{w}_{\tau} + \sum_{\tau=0}^{s}\left(F^\top\right)^{\tau} P\left(\widehat{w}_{\tau} - \widehat{w}_{\tau+\nt-s}\right).
\end{align*}
By our model assumption, $\|w_t\|\leq \omega$ and  $\|\widehat{w}_t\|\leq \overline{w}$ for all $t=0,\ldots,\nt-1$. Then, there exists some $e>0$ such that the prediction error $e_t=\widehat{w}_t-w_t$ satisfies $e_t\leq e$ for all $t=0,\ldots,\nt-1$.
Note that $F=A-BK$ and we define $\rho\coloneqq \frac{1+\rho(F)}{2}<1$ where  $\rho(F)$ denotes the spectral radius of $F$. From Gelfand’s formula, there exists a constant $C\geq 0$ such that $\|F^t\|\leq C\rho^{t}$ for all $t\geq 0$.
The following holds for $ \overline{\eta}(\widehat{w};s,\nt)$ and $ \overline{\eta}(w;s,\nt)$:
\begin{align}
\label{eq:4.13}
 \left\|\overline{\eta}(\widehat{w};s,\nt)\right\| \leq & \sum_{\tau=0}^{s}\left\|F^\tau\right\| \|P\|\overline{w} \leq C \frac{1-\rho^{s+1}}{1-\rho}\|P\|\overline{w}\leq \frac{C}{1-\rho}\|P\|\overline{w},\\
 \label{eq:4.14}
 \left\|\overline{\eta}(w;s,\nt)\right\| \leq & \sum_{\tau=0}^{s}\left\|F^{\tau}\right\| \|P\|\overline{w} = C \frac{1-\rho^{s+1}}{1-\rho}\|P\|\overline{w} \leq \frac{C}{1-\rho}\|P\|\overline{w}.
\end{align}
Moreover,
\begin{align}
\label{eq:4.15}
\left\|\xi_1\right\|\leq & \sum_{\tau=s+1}^{\nt}\left\|F^\tau \right\|  \|P\|\overline{w} + \sum_{\tau=0}^{s}\left\|F^\tau\right\|  \|P\|\left\|w_{\tau} - w_{\tau+\nt-s}\right\|\\ \leq & \frac{C \|P\|}{1-\rho}\left(\overline{w}\rho^{s+1} + \max_{\tau}\left\|w_{\tau} - w_{\tau+\nt-s}\right\|\right)\\
\left\|\xi_2\right\|\leq & \sum_{\tau=s+1}^{\nt}\left\|F^\tau\right\| \|P\|\widehat{w} + \sum_{\tau=0}^{s}\left\|F^\tau\right\|  \|P\|\left\|\widehat{w}_{\tau} - \widehat{w}_{\tau+\nt-s}\right\|\\
\label{eq:4.16}
\leq & \frac{C\|P\|}{1-\rho}\left(\widehat{w}\rho^{s+1} + \max_{\tau}\left\|\widehat{w}_{\tau} - \widehat{w}_{\tau+\nt-s}\right\|\right).
\end{align}
Combining~\eqref{eq:4.13}-\eqref{eq:4.16} with~\eqref{eq:4.12},
\begin{align*}
&\left| \overline{\eta}(w;\nt,\nt)^\top H \overline{\eta}(\widehat{w};\nt,\nt) -  \overline{\eta}(w;s,\nt)^\top H \overline{\eta}(\widehat{w};s,\nt)\right| \\
\leq &2\|H\|\left(\frac{C \|P\|}{1-\rho}\right)^2\left(2\rho^{s+1}\overline{w}\widehat{w} +\max_{\tau}\left\|\widehat{w}_{\tau} - \widehat{w}_{\tau+\nt-s}\right\| +\max_{\tau}\left\|w_{\tau} - w_{\tau+\nt-s}\right\|\right).
\end{align*}
\end{proof}

Therefore, applying Lemma~\ref{lemma:series_2}, we conclude that
\begin{align}
\nonumber
    (a) \coloneqq  & \left|\frac{1}{\nt}\sum_{s=0}^{\nt} \eta(w;s,\nt)^\top H \eta(\widehat{w};s,\nt) -\frac{1}{t}\sum_{s=0}^{t} \eta(w;s,\nt)^\top H \eta(\widehat{w};s,\nt)\right|\\
   \label{eq:bound_a} 
    \leq &\frac{4}{t}\|H\|\rho\left(\frac{C\|P\|}{(1-\rho)^{3/2}}\right)^2\overline{w}\widehat{w} + \frac{2}{t}\|H\|\left(\frac{C\|P\|}{1-\rho}\right)^2 \left(\mu_{\mathsf{VAR}}(\widehat{\mathbf{w}})+\mu_{\mathsf{VAR}}(\mathbf{w})\right)
\end{align}
where $    \mu_{\mathsf{VAR}}(\mathbf{x})\coloneqq \sum_{s=0}^{\nt}\max_{\tau}\left\|x_{\tau} - x_{\tau+\nt-s}\right\|$ denotes the self-variation of a sequence $\mathbf{x}$.

\subsubsection{Upper bound on (b)}

Next, we provide a bound on (b) in~\eqref{eq:centralbound}.
For (b), we have
\begin{align}
\nonumber
(b) \coloneqq  &\frac{1}{t}\left|\sum_{s=0}^{t} \eta(w;s,t)^\top H \eta(\widehat{w};s,t) -\sum_{s=0}^{t} \eta(w;s,\nt)^\top H \eta(\widehat{w};s,\nt)\right|\\
\nonumber
\leq &\frac{1}{t}\left|\sum_{s=0}^{t} \eta(w;s,\nt)^\top H \eta(\widehat{w};s,t) -\sum_{s=0}^{t} \eta(w;s,\nt)^\top H \eta(\widehat{w};s,\nt)\right|\\
\label{eq:4.21}
+& \frac{1}{t}\left|\sum_{s=0}^{t} \eta(w;s,\nt)^\top H \eta(\widehat{w};s,t) -\sum_{s=0}^{t} \eta(w;s,t)^\top H \eta(\widehat{w};s,t)\right|.
\end{align}
Noting that
$
\eta(\widehat{w};s,\nt)-\eta(\widehat{w};s,t)  = \sum_{\tau=t+1}^{\nt}\left(F^\top\right)^{\tau-s}P\widehat{w}_{\tau} 
$,
we obtain
\begin{align}
\nonumber
&\left|\sum_{s=0}^{t} \eta(w;s,\nt)^\top H \eta(\widehat{w};s,t) -\sum_{s=0}^{t} \eta(w;s,\nt)^\top H \eta(\widehat{w};s,\nt)\right|\\
\nonumber
= &\left|\sum_{s=0}^{t} \eta(w;s,\nt)^\top H \left(\eta(\widehat{w};s,t)- \eta(\widehat{w};s,\nt)\right)\right|\\
\label{eq:4.50}
=&\left|\sum_{s=0}^{t} \eta(w;s,\nt)^\top H \left( \sum_{\tau=t+1}^{\nt}\left(F^\top\right)^{\tau-s}P\widehat{w}_{\tau}\right)\right|
\end{align}
and similarly,
\begin{align}
\nonumber
&\left|\sum_{s=0}^{t} \eta(w;s,\nt)^\top H \eta(\widehat{w};s,t) -\sum_{s=0}^{t} \eta(w;s,t)^\top H \eta(\widehat{w};s,t)\right|\\
\nonumber
=&\left|\sum_{s=0}^{t} \left(\eta(w;s,\nt)-\eta(w;s,t)\right)^\top H \eta(\widehat{w};s,t)\right|\\
\label{eq:4.20}
=&\left|\sum_{s=0}^{t} \left( \sum_{\tau=t+1}^{\nt}\left(F^\top\right)^{\tau-s}Pw_{\tau}\right)^\top H \eta(\widehat{w};s,t) \right|.
\end{align}
By our assumption, $\|w_t\|\leq \overline{w}$ and  $\|\widehat{w}_t\|\leq \widehat{w}$ for all $t=0,\ldots,\nt-1$. Therefore, for any $s\leq t$:
\begin{align}
   \label{eq:4.8}
   \left\| \sum_{\tau=t+1}^{\nt}\left(F^\top\right)^{\tau-s}P\widehat{w}_{\tau}\right\|
   \leq & \frac{C\rho^{t-s+1}\|P\|\widehat{w}}{1-\rho}
\end{align}
and
\begin{align}
  \label{eq:4.9}
  \left\|\eta(w;s,\nt)\right\| = & \left\|\sum_{\tau=s}^{\nt}\left(F^{\top}\right)^{\tau-s}Pw_\tau\right\|
  \leq \frac{{C \|P\|}\overline{w}}{1-\rho}.
\end{align}
 Plugging~\eqref{eq:4.8} and~\eqref{eq:4.9} into~\eqref{eq:4.50}, 
\begin{align}
\nonumber
&\left|\sum_{s=0}^{t} \eta(w;s,\nt)^\top H \eta(\widehat{w};s,t) -\sum_{s=0}^{t} \eta(w;s,\nt)^\top H \eta(\widehat{w};s,\nt)\right|\\
\nonumber
\leq &2C\|H\|\left(\frac{\|P\|}{1-\rho}\right)^2 \overline{w}\widehat{w}\sum_{s=0}^{t}\left\|F^{t-s+1}\right\|\\
\nonumber
\leq &2\|H\|\left(\frac{C\|P\|}{1-\rho}\right)^2 \frac{\rho\left(1-\rho^t\right)}{1-\rho}\overline{w}\widehat{w}\\
\label{eq:4.22}
\leq &2\|H\|\left(\frac{ C \|P\|}{(1-\rho)^{3/2}}\right)^2 {\rho}\overline{w}\widehat{w}.
\end{align}
Using the same argument, the following bound holds for~\eqref{eq:4.20}:
\begin{align}
\label{eq:4.23}
&\left|\sum_{s=0}^{t} \eta(w;s,\nt)^\top H \eta(\widehat{w};s,t) -\sum_{s=0}^{t} \eta(w;s,t)^\top H \eta(\widehat{w};s,t)\right|
\leq 2\|H\|\left(\frac{C\|P\|}{(1-\rho)^{3/2}}\right)^2 {\rho}\overline{w}\widehat{w}.
\end{align}
Combining~\eqref{eq:4.22} and~\eqref{eq:4.23} and using~\eqref{eq:4.21},
\begin{align}
\nonumber
    (b)\coloneqq  & \frac{1}{t}\left|\sum_{s=0}^{t} \eta(w;s,t)^\top H \eta(\widehat{w};s,t) -\sum_{s=0}^{t} \eta(w;s,\nt)^\top H \eta(\widehat{w};s,\nt)\right| \\
    \label{eq:bound_b}
    \leq & \frac{4}{t}\|H\|\left(\frac{C\|P\|}{(1-\rho)^{3/2}}\right)^2 {\rho}\overline{w} \widehat{w}.
\end{align}

Finally, together,~\eqref{eq:bound_a} and~\eqref{eq:bound_b} imply the following:
\begin{align}
\nonumber
      \left|\frac{1}{\nt}V(\nt)-\frac{1}{t}V(t)\right| \leq 
      \frac{8}{t} &\|H\|\left(\frac{C\|P\|}{(1-\rho)^{3/2}}\right)^2 {\rho}\overline{w} \widehat{w} \\
      \label{eq:2.25}
       + \frac{2}{t} &\|H\|\left(\frac{C\|P\|}{1-\rho}\right)^2 \left(\mu_{\mathsf{VAR}}(\mathbf{w})+\mu_{\mathsf{VAR}}(\mathbf{\widehat{w}})\right).
\end{align}
The same argument also guarantees that
\begin{align}
\nonumber
      \left|\frac{1}{\nt}W(\nt)-\frac{1}{t}W(t)\right| \leq 
      \frac{8}{t} &\|H\|\left(\frac{C\|P\|}{(1-\rho)^{3/2}}\right)^2 {\rho}\widehat{w}^2 \\
      \label{eq:2.26}
       + \frac{4}{t} &\|H\|\left(\frac{C\|P\|}{1-\rho}\right)^2 \mu_{\mathsf{VAR}}(\mathbf{\widehat{w}}).
\end{align}
The following lemma together with~\eqref{eq:2.25} and~\eqref{eq:2.26} justify the conditions needed to apply Lemma~\ref{lemma:series_1}. 
\begin{lemma}
\label{lemma:extra_conditions}
For any integer $1\leq t\leq \nt$,
\begin{align*}
    \frac{V(t)}{t}\leq & 2\|H\|\left(\frac{C\|P\|}{(1-\rho)^{3/2}}\right)^2\overline{w}\widehat{w},
\end{align*}
where $C>0$ is some constant satisfying $\|F^t\|\leq C\rho^{t}$ for all $t\geq 0$.
\end{lemma}

\begin{proof}[Proof of Lemma~\ref{lemma:extra_conditions}]
We have
\begin{align*}
    \frac{V(t)}{t}=&\frac{1}{t}\sum_{s=0}^{t}\eta(w;s,t)^\top H {\eta}(\widehat{w};s,t)\\
    \leq & \frac{\|H\|}{t}\sum_{s=0}^{t}\left\|\sum_{\tau=0}^{t-1-s}\left(F^{\top}\right)^{\tau}Pw_{\tau+s}\right\|\left\|\sum_{\tau=0}^{t-1-s}\left(F^{\top}\right)^{\tau}P\widehat{w}_{\tau+s}\right\|\\
    \leq & \frac{\|H\|}{t}\left(\frac{C\|P\|}{1-\rho}\right)^2\sum_{s=0}^{t}\left(1-\rho^{t-s}\right) \overline{w}\widehat{w}\\
    = &\frac{\|H\|}{t}\left(\frac{C\|P\|}{1-\rho}\right)^2\left(t+\frac{1-\rho^{t+1}}{1-\rho}\right) \overline{w}\widehat{w}\\
    \leq &2\|H\|\left(\frac{C\|P\|}{(1-\rho)^{3/2}}\right)^2\overline{w}\widehat{w}.
\end{align*}
\end{proof}

First, based on our assumption, $\lambda_{t}=V(t)/W(t)=V_t/W_t\leq 1$. Moreover, $W(\nt)/\nt=\Omega(1)$.
Therefore, using~\eqref{eq:2.25}, ~\eqref{eq:2.26}, Lemma~\ref{lemma:series_1} and Lemma~\ref{lemma:extra_conditions},~\eqref{eq:4.6} implies that for any $1<t\leq\nt$,
\begin{align*}
\left|\lambda_t - \lambda_{\nt} \right|\leq    \frac{1}{t-1}\frac{\|H\|\left(\frac{C\|P\|}{1-\rho}\right)^2}{W(\nt)/\nt}\cdot&\Bigg(\frac{8\rho\widehat{w} \overline{w}}{1-\rho}+2\left(\mu_{\mathsf{VAR}}(\mathbf{w})+\mu_{\mathsf{VAR}}(\mathbf{\widehat{w}})\right)\\
+\frac{2\|H\|\left(\frac{C\|P\|}{(1-\rho)^{3/2}}\right)^2\overline{w}\widehat{w}}{W(\nt)/\nt}&\left(\frac{8\rho\widehat{w}^2}{1-\rho}+4\mu_{\mathsf{VAR}}(\mathbf{\widehat{w}})\right)\Bigg)\\
=O\left(\frac{\mu_{\mathsf{VAR}}(\mathbf{w})+\mu_{\mathsf{VAR}}(\mathbf{\widehat{w}})}{t}\right).
\end{align*}

\subsection{Proof of Lemma~\ref{lemma:regret}}
\label{appendix: proof_of_lemma:regret}
Using Lemma~\ref{lemma:bound_eta},
\begin{align*}
    |\lambda_t-\lambda_{\nt}|\leq \frac{C_1}{t}\left(\mu_{\mathsf{VAR}}(\mathbf{w})+\mu_{\mathsf{VAR}}(\mathbf{\widehat{w}})\right), \quad \text{ where } C_1>0 \text{ is } \text{some  constant}.
\end{align*}
Applying Lemma~\ref{lemma:regret_bound}, and noting that
\begin{align*}
\left\|\sum_{\tau=t}^{\nt-1}\left(F^\top\right)^{\tau-t}P\widehat{w}_\tau\right\|\leq C\frac{1-\rho^{\nt-t}}{1-\rho}\|P\|\widehat{w},
\end{align*}
\eqref{eq:4.7} implies 
\begin{align}
\nonumber
\mathsf{Regret}
 \leq & 2C_1\|H\|\sum_{t=1}^{\nt-1}\left\|\frac{\mu_{\mathsf{VAR}}(\mathbf{w})+\mu_{\mathsf{VAR}}({\mathbf{\widehat{w}}})}{t}\sum_{\tau=t}^{\nt-1}\left(F^\top\right)^{\tau-t}P\widehat{w}_\tau\right\| + C_0 \\
 \nonumber
= & 2C_1\|H\|\left(\mu_{\mathsf{VAR}}(\mathbf{w})+\mu_{\mathsf{VAR}}({\mathbf{\widehat{w}}})\right)\sum_{t=1}^{\nt-1}\frac{1}{t}\left\|\sum_{\tau=t}^{\nt-1}\left(F^\top\right)^{\tau-t}P\widehat{w}_\tau\right\| + C_0\\
\nonumber
\leq & 2C_1\|H\|\left(\mu_{\mathsf{VAR}}(\mathbf{w})+\mu_{\mathsf{VAR}}({\mathbf{\widehat{w}}})\right)\left(\frac{C\|P\|}{1-\rho}\widehat{w}\right)\sum_{t=1}^{\nt-1}\frac{1}{t} + C_0 \\
 \label{eq:4.45}
\leq & 2C_1\|H\|\left(\frac{C\|P\|}{1-\rho}\widehat{w}\right)\left(\mu_{\mathsf{VAR}}(\mathbf{w})+\mu_{\mathsf{VAR}}({\mathbf{\widehat{w}}})\right)(1+\log\nt) + C_0 
\end{align}
where 
\begin{align*}
 C_0\coloneqq  & \|H\|\left|\lambda_{\nt}+\lambda_0\right|\left|\lambda_{\nt}-\lambda_0\right|\left\|\sum_{\tau=0}^{\nt-1}\left(F^\top\right)^{\tau}P\widehat{w}_\tau\right\| \leq \|H\|\left|\lambda_{\nt}+\lambda_0\right|\left|\lambda_{\nt}-\lambda_0\right|\left(\frac{C\|P\|}{{1-\rho}} \widehat{w}\right).
\end{align*}
Moreover, for any $t=1,\ldots,\nt$, $|\lambda_t| \leq 1$,
whence,
\begin{align*}
C_0\leq 2\|H\|\left(\frac{C\|P\|}{{1-\rho}} \widehat{w}\right).
\end{align*}
Therefore, continuing from~\eqref{eq:4.45},
\begin{align*}
    \mathsf{Regret}\leq &\|H\|\left(\frac{C\|P\|}{{1-\rho}} \widehat{w}\right)\Big(2C_1\left(\mu_{\mathsf{VAR}}(\mathbf{w})+\mu_{\mathsf{VAR}}({\mathbf{\widehat{w}}})\right)(1+\log \nt)+2\Big)\\
    = &O\Big(\big(\mu_{\mathsf{VAR}}(\mathbf{w})+\mu_{\mathsf{VAR}}({\mathbf{\widehat{w}}})\big)\log\nt\Big).
\end{align*}









\section{Proof of Theorem~\ref{thm:threshold}}
\label{appendix:proof_of_threshold}
First, note that the total cost is given by
$
    J=\sum_{t=0}^{T-1}x_t^\top Qx_t+u_t^\top Ru_t + x_T^\top Px_T.
$
Since we can choose a threshold $\sigma>0$ arbitrarily small, the error must exceed a threshold $\sigma$. Without loss of generality, we suppose that the accumulated error $\delta$ exceeds the threshold $\sigma$ at time $s\geq 0$ and assume the predictions $\widehat{w}_t$, $0<t<s-1$ are accurate. 

Throughout, we define $J_1 \coloneqq  \sum_{t=1}^{s-1}x_t^\top Qx_t+u_t^\top Ru_t$ and $J_2 \coloneqq  \sum_{t=s}^{T-1}x_t^\top Qx_t+u_t^\top Ru_t$
and use diacritical letters $\widehat{J},\widehat{x}$ and $\widehat{u}$ to denote the corresponding cost, action and state of the threshold algorithm (Algorithm~\ref{alg:threshold}). 
We consider the best online algorithm (with no predictions available) that minimizes its corresponding competitive ratio and use diacritical letters $\widetilde{J},\widetilde{x}$ and $\widetilde{u}$ to denote the corresponding cost, action and state. The competitive ratio of the best online algorithm is denoted by $C_{\min}$.

\subsection{Upper Bound on $\widehat{J}_1$}\label{appendix:upperbound_J1}

We first provide an upper bound on $\widehat{J}_1$, the first portion of the total cost.
For $1\le t< s$, the threshold-based algorithm gives
\begin{align*}
    \widehat{u}_t&= -K \widehat{x}_t -(R+B^{\top}PB)^{-1}B^{\top}\left(\sum_{\tau=t}^{\nt-1}\left(F^\top\right)^{\tau-t} P\widehat{w}_{\tau}\right)\\
    &=-K \widehat{x}_t -(R+B^{\top}PB)^{-1}B^{\top}\left(\sum_{\tau=t}^{\nt-1}\left(F^\top\right)^{\tau-t} P w_{\tau} - \eta_t\right).
\end{align*}
Lemma $10$ in \cite{yu2020competitive} implies 
\begin{equation}
    \nonumber
    J_1 = \mathsf{ALG}(0:T)-\mathsf{ALG}(s:T)
\end{equation}
where 
\begin{align}
\nonumber
    \mathsf{ALG}(0:T)=&\sum_{t=0}^{T-1}\left(w_t^{\top} P w_t+2 w_t^{\top}\sum_{i=1}^{T-t-1}\left(F^{\top}\right)^i P w_{t+i}\right)\\
    \nonumber
    &-\sum_{t=0}^{T-1}\left(\sum_{i=0}^{T-t-1}\left(F^{\top}\right)^i P w_{t+i}\right)^{\top} H \left(\sum_{i=0}^{T-t-1}\left(F^{\top}\right)^i P w_{t+i}\right)\\
    \label{eq:5.3}
    &+\sum_{t=0}^{T-1}\eta_t^{\top} H \eta_t+x_0^{\top} P x_0+ 2x_0^{\top}\sum_{i=0}^{T-1}\left(F^{\top}\right)^{i+1}P w_i,
\end{align}
and
\begin{align}
\nonumber
   \mathsf{ALG}(s:T)
    \coloneqq&\sum_{t=0}^{T-s-1}\left(w_{t+s}^{\top} P w_{t+s}+2 w_{t+s}^{\top}\sum_{i=1}^{T-s-t-1}\left(F^{\top}\right)^i P w_{t+s+i}\right)\\
    \nonumber
    &-\sum_{t=0}^{T-s-1}\left(\sum_{i=0}^{T-s-t-1}\left(F^{\top}\right)^i P w_{t+s+i}\right)^{\top} H \left(\sum_{i=0}^{T-s-t-1}\left(F^{\top}\right)^i P w_{t+s+i}\right)\\
    \label{eqq:5.0}
    &+\sum_{t=0}^{T-s-1}\eta_{t+s}^{\top} H \eta_{t+s}+x_s^{\top} P x_s+ 2x_s^{\top}\sum_{i=0}^{T-s-1}\left(F^{\top}\right)^{i+1}P w_{i+s}.
\end{align}
Rewriting~\eqref{eqq:5.0},
\begin{align}
\nonumber
   \mathsf{ALG}(s:T)
    \coloneqq &
    \sum_{t=s}^{T-1}\left(w_t^{\top} P w_t+2 w_t^{\top}\sum_{i=1}^{T-t-1}\left(F^{\top}\right)^i P w_{t+i}\right)\\
    \nonumber
    &-\sum_{t=s}^{T-1}\left(\sum_{i=0}^{T-t-1}\left(F^{\top}\right)^i P w_{t+i}\right)^{\top} H \left(\sum_{i=0}^{T-t-1}\left(F^{\top}\right)^i P w_{t+i}\right)\\
    \label{eq:5.2}
    &+\sum_{t=s}^{T-1}\eta_t^{\top} H \eta_t+x_s^{\top} P x_s+ 2x_s^{\top}\sum_{i=s}^{T-1}\left(F^{\top}\right)^{i+1-s}P w_i.
\end{align}

Therefore, combining~\eqref{eq:5.3} and~\eqref{eq:5.2},
\begin{align*}
    J_1=&\sum_{t=0}^{s-1}\left(w_t^{\top} P w_t+2 w_t^{\top}\sum_{i=1}^{T-t-1}\left(F^{\top}\right)^i P w_{t+i}\right)\\
    &-\sum_{t=0}^{s-1}\left(\sum_{i=0}^{T-t-1}\left(F^{\top}\right)^i P w_{t+i}\right)^{\top} H \left(\sum_{i=0}^{T-t-1}\left(F^{\top}\right)^i P w_{t+i}\right)\\
    &+\sum_{t=0}^{s-1}\eta_t^{\top} H \eta_t+x_0^{\top} P x_0+ 2x_0^{\top}\sum_{i=0}^{T-1}\left(F^{\top}\right)^{i+1}P w_i-x_s^{\top} P x_s- 2x_s^{\top}\sum_{i=s}^{T-1}\left(F^{\top}\right)^{i+1-s}P w_i.
\end{align*}
Denote by $\Delta J_1\coloneqq \left|J_1-\widehat{J}_1\right|$. We obtain
\begin{align*}
    \Delta J_1=& \sum_{t=0}^{s-1}\eta_t^{\top} H \eta_t+x_s^{\top} P x_s - \widehat{x}_s^{\top} P \widehat{x}_s +2(x_s-\widehat{x}_s)^{\top}\sum_{i=s}^{T-1}\left(F^{\top}\right)^{i+1-s}P w_i\\
    =&\sum_{t=0}^{s-1}\eta_s^{\top} F^{s-t} H\left((F^{\top})^{s-t}\eta_s\right)+x_s^{\top} P x_s - \widehat{x}_s^{\top} P \widehat{x}_s +2(x_s-\widehat{x}_s)^{\top}\sum_{i=s}^{T-1}\left(F^{\top}\right)^{i+1-s}P w_i\\
    \le&\frac{c\|H\|}{1-\rho^2}\frac{c^2\|P\|^2R^2}{(1-\rho)^2}+2\|P\|\|x_s\|\|x_s-\widehat{x}_s\|+\|x_s-\widehat{x}_s\|^2+2\|x_s-\widehat{x}_s\|\frac{c\|P\|\rho}{1-\rho}.
\end{align*}
Since the following is true:
\begin{align*}
    x_s-\widehat{x}_s=&A(x_{s-1}-\widehat{x}_{s-1})+B(u_{s-1}-\widehat{u}_{s-1})\\
    =&(A-BK)(x_{s-1}-\widehat{x}_{s-1})+B(R+B^{\top} P B)^{-1}B^{\top}\eta_{s-1}\\
    =&\sum_{t=0}^{s-1}(F^{\top})^{s-t-1}B(R+B^{\top} P B)^{-1}B^{\top}\eta_{t},
\end{align*}
we have 
\begin{align*}
\nonumber
    \|x_s-\widehat{x}_s\|\le&\frac{c^2\|B(R+B^{\top} P B)^{-1}B^{\top}\|R}{(1-\rho)^2}.
\end{align*}
If $\|x_s\|=O(1)$, then $\Delta J_1 = O(1)$, else 
\begin{equation}
\nonumber
    \frac{\Delta J_1}{J_1}\le\frac{O(1)\cdot \|x_s\|+O(1)}{x_s^{\top} Q x_s}\to 0.
\end{equation}
Therefore, as a conclusion, $\widehat{J}_1$ can be bounded from above by
\begin{equation}\label{eq:D4.1}
    \widehat{J}_1\le J_1+O(1).
\end{equation}

\subsection{Upper Bound on $\widehat{J}_2$}

For section~\ref{appendix:upperbound_J1}, we know that $\|x_s-\widehat{x}_s\|=O(1)$. Let $\widetilde{J}_2$ denote the cost by running 1-confident algorithm from $\widehat{x}_s$ with correct prediction, and $\widetilde{x}_t$ denote the state we get in the procedure. Then 
\begin{equation}
\nonumber
    \|x_t-\widetilde{x}_t\|=\|(A-BK)(x_{t-1}-\widetilde{x}_{t-1})\|=\|F^{t-s}(x_s-\widetilde{x}_s)\|=\|F^{t-s}(x_s-\widehat{x}_s)\|.
\end{equation}
Therefore, \begin{align*}
    |J_2-\widetilde{J}_2|\le&\left|\sum_{t=s}^{T-1}\left(\widetilde{x}_t-x_t\right)^{\top} Q x_t +x_t^{\top} Q\left(\widetilde{x}_t-x_t\right) + \left(\widetilde{x}_t-x_t\right)^{\top} Q\left(\widetilde{x}_t-x_t\right)\right|\\
    &+|\sum_{t=s}^{T-1}\left(\widetilde{u}_t-u_t\right)^{\top} R u_t + u_t^{\top} R \left(\widetilde{u}_t-u_t\right) + \left(\widetilde{u}_t-u_t\right)^{\top} R \left(\widetilde{u}_t-u_t\right)|\\
    &+\left|\left(\widetilde{x}_T-x_T\right)^{\top} P x_T + x_T^{\top} P \left(\widetilde{x}_T-x_T\right) + \left(\widetilde{x}_T-x_T\right)^{\top} P \left(\widetilde{x}_T-x_T\right)\right|\\
    \le&\sum_{t=s}^{T-1}\left(\|Q\|+\|K^{\top} R K\|\right)\left\|F^{2t-2s}\right\|\left\|x_s-\widehat{x}_s\right\|^2\\
    &+\sum_{t=s}^{T-1}2\left\|F^{t-s}\right\|\left\|x_s-\widehat{x}_s\right\|(\|Q\|\|x_t\|+\|RK\|\|u_t\|)\\
    &+2\left\|F^{T-s}\right\|\|P\|\|\|x_s-\widehat{x}_s\|\|x_T\|+\left\|F^{2T-2s}\right\|\|P\|\|x_s-\widehat{x}_s\|^2\\
    =& \sum_{t=s}^{T-1}2\left\|F^{t-s}\right\|\left\|x_s-\widehat{x}_s\right\|(\|Q\|\|x_t\|+\|RK\|\|u_t\|) \\
    &b+ 2\left\|F^{T-s}\right\|\|P\|\|\left\|x_s-\widehat{x}_s\right\|\|x_T\| + O(1).
\end{align*}
If $\|x_t\|=O(1)$ and $\|u_t\|=O(1)$ for all $t$, then \begin{align*}
    \left|J_2-\widetilde{J}_2\right|=O(1).
\end{align*} 
Otherwise, suppose $x_{i_1}, x_{i_2}, \ldots, x_{i_k}$ and $u_{j_1}, u_{j_2},\ldots, u_{j_l}$ are some functions of $T$, then for any $1\le m\le k$ and $1\le n\le l$, ${\|x_{i_m}\|}/{x_{i_m}^{\top} Q x_{i_m}}\to 0$ and ${\|u_{j_n}\|}/{u_{j_n}^{\top} R u_{j_n}}\to 0$. Therefore, \begin{align*}
\frac{\left|J_2-\widetilde{J}_2\right|}{J_2}
    \le &2\|x_s-\widehat{x}_s\|\frac{\sum_{m=1}^k\left\|F^{i_m-s}\right\|\|Q\|\|x_{i_m}\|+\sum_{n=1}^l\left\|F^{j_n-s}\right\|\|R K\|\|u_{j_n}\|}{J_2}+\frac{O(1)}{J_2}\to 0.
\end{align*}
Combining the two cases, we can conclude that 
\begin{equation}\label{eq:D4.2}
    \left|J_2-\widetilde{J}_2\right|\le J_2+O(1).
\end{equation}

Therefore, from \eqref{eq:D4.1} and \eqref{eq:D4.2}, we conclude that
\begin{equation}
\widehat{J}=\widehat{J}_1+\widehat{J}_2\le J_1+O(1)+C_{\min}\widetilde{J}_2\le J_1+O(1)+C_{\min}(J_2+O(1))=C_{\min}J+O(1).
\end{equation}
The proof completes by noticing that when the prediction error is zero and $\widehat{w}_t=w_t$ for all $t=0,\ldots,\nt-1$, the accumulated error $\delta$ will always be $0$ and since the threshold $\sigma$ is positive, the algorithm is always optimal and $1$-consistent. As a result, Algorithm~\ref{alg:threshold} is $1$-consistent and $(C_{\min}+o(1))$-robust.

\section{Experimental Setup}
\label{appendix:experiments}

In our three case studies, we consider i.i.d.\ prediction errors, i.e., $e_t=\widehat{w}_t - w_t$ is an i.i.d. additive prediction noise. To illustrate the effects of randomness for simulating the worst-case performance, we consider varying types of noise in the case studies. For the robot tracking case, we set $e_t=c X$ where $X\sim B(10,0.5)$ is a binomial random variable with $10$ trials and a success probability $0.5$ and $c>0$ is a scaling parameter. For the battery-buffered EV charging case, we set $e_t= Y$ where $X\sim N(0,\sigma^2)$ is a normal random variable with zero mean and $\sigma^2$ is a variance that can be varied to generate varying prediction error. For the Cart-Pole problem, we set $e_t=Z w_t$ where $w_t=60\times B$ with $\eta\coloneqq l\left(\frac{4}{3}-\frac{m}{m+M}\right)$,
\begin{align*}
  B\coloneqq
\begin{bmatrix}
0 \\ \frac{(m+M)\eta+m l}{(m+M)^2\eta} \\ 0 \\ -\frac{1}{(m+M)\eta},
\end{bmatrix}  
\end{align*}
and $Z\sim N(0,\sigma^2)$ is a normal random variable with zero mean and $\sigma^2$ is a variance ranging between $0$ to $8\times 10 ^2$.
To simulate the worst-case performance of algorithms, in our experiments we run the algorithms $5$ times, with a new sequence of prediction noise generated at each time and choose the one with the largest overall cost.

Finally, Table~\ref{table:parameter} and Table~\ref{table:parameter2} list the detailed settings and the hyper-parameters used in the robot tracking, battery-buffered EV charging and Cart-Pole case studies.


\begin{table}[h]
\footnotesize
\renewcommand{\arraystretch}{1.3}
    \centering
       \caption{Hyper-parameters used in robot tracking and EV charging.}
    \begin{tabular}{l|l|| l|l}
    \hline
   Robot Tracking  & Value &    EV Charging & Value \\
   \hline
   Number of Monte Carlo Tests & 5 & Number of Monte Carlo Tests& 5\\
    Prediction Error Type & \textit{Binomial} &  Prediction Error Type & \textit{Gaussian}\\
    State Dimension $\nn$ & $4$ & State Dimension  $\nn$ & $10$ (Synthetic); $52$ (Realistic) \\
    Action Dimension $\nm$ & $2$ & Action Dimension $\nm$ & $10$ (Synthetic); $52$ (Realistic)\\
    Time Horizon Length $\nt$ & Fig~\ref{fig:binomial_convergence_trajectory}: $\nt=240$ & Time Horizon Length $\nt$  & $240$\\
    & Fig~\ref{fig:gaussian_convergence_trajectory}: $\nt=240$  &   & \\
    & Fig~\ref{fig:impact_tracking}: $\nt=200$ &   & \\
    Initialized $\lambda_0$  & $0.3$ & Charging Efficiency & $1$\\
    Scaling parameter $c$ & Fig~\ref{fig:impact_tracking}: $c\in [0,1]$ & Variance $\sigma^2$ & $\sigma^2\in [0,10]$ \\
    CPU  & Intel{\textregistered}  i7-8850H & CPU & Intel{\textregistered}   i7-8850H \\
    & & Energy Demand $E$ (Synthetic Case) & 5 (kWh) \\
    & & Arrival Rate (Realistic Case)  & 0.2 \\
       \hline
       \hline
    \end{tabular}
  \label{table:parameter}
\end{table}

\begin{table}[h]
\footnotesize
\renewcommand{\arraystretch}{1.3}
    \centering
       \caption{Hyper-parameters used in the Cart-Pole problem.}
    \begin{tabular}{l|l}
    \hline
   Robot Tracking  & Value  \\
   \hline
   Number of Monte Carlo Tests & 2000 \\
    Prediction Error Type & \textit{Gaussian} \\
    Action Dimension $\nm$ & 1\\
    State Dimension $\nm$ & 4\\
    Time Horizon Length $\nt$ & $200$\\
    Variance $\sigma^2$ &  $\sigma^2\in [0,800]$  \\
    Cart Mass $M$ &  $M=10.0 kg$\\
    Pole Mass $m$ &  $m=1.0 kg$\\
    Pole Length $l$ & $l=10.0 m$ \\
    CPU  & Intel{\textregistered}  i7-8850H \\
       \hline
       \hline
    \end{tabular}
  \label{table:parameter2}
\end{table}

\end{document}